\documentclass[11pt,letterpaper]{article}
\usepackage[margin=1.04in]{geometry}
\usepackage{xspace}
\usepackage{varioref}
\usepackage{placeins}
\usepackage[utf8x]{inputenc}
\usepackage{listofitems}
\setsepchar{ }
\usepackage[dvipsnames]{xcolor}
\usepackage{tikz}
\usepackage{soul}
\usepackage{times}
\usepackage{rotating}
\usetikzlibrary{positioning}
\usepackage{xcolor}
\usepackage{setspace}
\xdefinecolor{green}{rgb}{0.2,0.8,0.2}
\xdefinecolor{red}{rgb}{0.7,0.1,0.1}
\usepackage{subcaption}
\usepackage{csquotes}
\usepackage{chngcntr}
\usetikzlibrary{calc, shapes, backgrounds}
\usetikzlibrary{shapes,shadows,arrows,positioning,graphs}
\usetikzlibrary{arrows,shapes,positioning}
\usetikzlibrary{decorations.markings}
\usetikzlibrary{decorations.pathreplacing}
\tikzstyle arrowstyle=[scale=1]
\tikzstyle directed=[postaction={decorate,decoration={markings,
    mark=at position .65 with {\arrow[arrowstyle]{stealth}}}}]
\tikzstyle reverse directed=[postaction={decorate,decoration={markings,
    mark=at position .3 with {\arrowreversed[arrowstyle]{stealth};}}}]
    \tikzstyle reverse directedd=[postaction={decorate,decoration={markings,
    mark=at position .65 with {\arrowreversed[arrowstyle]{stealth};}}}]
\tikzset{defnode/.style={draw, circle, fill=white, inner sep=0, minimum size=5pt}}
\tikzset{invisible/.style={draw = white, circle, fill=white, inner sep=0, minimum size=6pt}}
\tikzset{emitter/.style={draw, circle, gray,very thin, fill=green, inner sep=0, minimum size=6pt}}
\tikzset{dism/.style={draw, circle, fill=gray, inner sep=0, minimum size=6pt}}
\tikzset{receiver/.style={draw, circle, gray,very thin, fill=red, inner sep=0, minimum size=6pt}}
\tikzset{emitter_receiver/.style={shape = circle split, draw, 
          minimum size = 8pt, inner sep = 0,
          rotate=30}}
\usepackage{stmaryrd}
\usepackage[toc,page]{appendix}
\usepackage{float}

\usepackage{amsthm,amssymb,amsmath}
\newtheorem{theorem}{Theorem}
\newtheorem{observation}{Observation}

\newtheorem{conjecture}{Open question}

\newtheorem{lemma}{Lemma}

\newtheorem{remark}{Remark}

\theoremstyle{definition}
\newtheorem{case}{Case}
\newtheorem{subcase}{Subcase}
\numberwithin{subcase}{case}

\usepackage{ucs} 
\usepackage{tkz-tab}
\usepackage{caption}
\usepackage{latexsym}
\usepackage{amssymb}
\usepackage{amsmath}
\usepackage{subcaption}

\newcommand{\ie}{{\it i.e.,}\xspace}
\newcommand{\eg}{{\it e.g.,}\xspace}
\newcommand{\m}{\ensuremath{{n \choose 2}}}

\newcommand{\G}{{\ensuremath{\cal G}}}
\newcommand{\HH}{{\ensuremath{\cal H}}}

\newcommand{\hexa}{
  \begin{tikzpicture}[scale=3.8]
    \path (0,0) coordinate (m);
    \path (m)+(90:1) node[defnode] (a){};
    \path (m)+(30:1) node[defnode] (b){};
    \path (m)+(330:1) node[defnode] (c){};
    \path (m)+(270:1) node[defnode] (d){};
    \path (m)+(210:1) node[defnode] (e){};
    \path (m)+(150:1) node[defnode] (f){};
    \tikzstyle{every node}=[circle,fill=white,inner sep=1.3pt]
    \ifnum\arg[1]>-1 %
    \draw (a)-- node[pos=.5]{\arg[1]}(b);
    \fi
    \ifnum\arg[2]>-1 %
    \draw (a)-- node[pos=.42]{\arg[2]}(c);
    \fi
    \ifnum\arg[3]>-1 %
    \draw (a)-- node[pos=.38]{\arg[3]}(d);
    \if
    \ifnum\arg[4]>-1 %
    \draw (a)-- node[pos=.58]{\arg[4]}(e);
    \fi
    \ifnum\arg[5]>-1 %
    \draw (a)-- node[pos=.5]{\arg[5]}(f);
    \fi
    \ifnum\arg[6]>-1 %
    \draw (b)-- node[pos=.5]{\arg[6]}(c);
    \fi
    \ifnum\arg[7]>-1 %
    \draw (b)-- node[pos=.42]{\arg[7]}(d);
    \fi
    \ifnum\arg[8]>-1 %
    \draw (b)-- node[pos=.38]{\arg[8]}(e);
    \fi
    \ifnum\arg[9]>-1 %
    \draw (b)-- node[pos=.58]{\arg[9]}(f);
    \fi
    \ifnum\arg[10]>-1 %
    \draw (c)-- node[pos=.5]{\arg[10]}(d);
    \fi
    \ifnum\arg[11]>-1 %
    \draw (c)-- node[pos=.42]{\arg[11]}(e);
    \fi
    \ifnum\arg[12]>-1 %
    \draw (c)-- node[pos=.38]{\arg[12]}(f);
    \fi
    \ifnum\arg[13]>-1 %
    \draw (d)-- node[pos=.5]{\arg[13]}(e);
    \fi
    \ifnum\arg[14]>-1 %
    \draw (d)-- node[pos=.42]{\arg[14]}(f);
    \fi
    \ifnum\arg[15]>-1 %
    \draw (e)-- node[pos=.5]{\arg[15]}(f);
    \fi
  \end{tikzpicture}
}
\usepackage[noabbrev,capitalize]{cleveref}

\begin{document}

\title{Temporal Cliques Admit Sparse Spanners
  \thanks{This article is the full version of an article presented at ICALP 2019~\cite{CPS19}. This paper is best read in colour. If this is not possible, the colour ``green'' appears with a lighter tone than the colour ``red'', making their interpretation possible.}
  \thanks{This research was supported by ANR project ESTATE (ANR-16-CE25-0009-03) and NSERC of Canada.}
}
\author{Arnaud Casteigts\\
LaBRI, University of Bordeaux, France\\
\texttt{arnaud.casteigts@labri.fr}
\and
Joseph G. Peters\\
School of Computing Science\\
Simon Fraser University,
Burnaby, BC, Canada\\
\texttt{peters@cs.sfu.ca}
\and
Jason Schoeters\\
LaBRI, University of Bordeaux, France\\
\texttt{jason.schoeters@labri.fr}
}

\date{ }
\maketitle

\begin{abstract}
Let $G=(V,E)$ be an undirected graph on $n$ vertices and $\lambda:E\to 2^{\mathbb{N}}$ a mapping that assigns to every edge a non-empty set of integer labels (discrete times when the edge is present). Such a labelled graph $\G=(G,\lambda)$ is {\em temporally connected} if a path exists with non-decreasing times from every vertex to every other vertex. In a seminal paper, Kempe, Kleinberg, and Kumar~\cite{KKK02} asked whether, given such a temporally connected graph, a {\em sparse} subset of edges always exists whose labels suffice to preserve temporal connectivity -- a {\em temporal spanner}. Axiotis and Fotakis~\cite{AF16} answered negatively by exhibiting a family of $\Theta(n^2)$-dense temporal graphs which admit no temporal spanner of density $o(n^2)$. In this paper, we give the first positive answer as to the existence of $o(n^2)$-sparse spanners in a dense class of temporal graphs, by showing (constructively) that if $G$ is a complete graph, then one can always find a temporal spanner with $O(n \log n)$ edges.
 \end{abstract}

\section{Introduction}
\label{sec:introduction}

The study of highly dynamic networks has gained interest lately, motivated by epidemics analysis and emerging technological contexts like vehicular networks, robots, and drones, where the entities move and communicate with each other. The communication links in these networks vary with time, leading to the definition of
temporal graph models (also called time-varying graphs or evolving graphs) where temporality plays a key role. In these graphs, 
the properties of interest are often defined over time rather than at a given instant. For example, the graph may never be connected, and yet offer a form of connectivity over time.
In~\cite{CFQS12}, a dozen temporal properties were identified that have been effectively exploited in the distributed computing and networking literature. Perhaps the most basic property is that of {\em temporal connectivity}, which requires that every vertex can reach every other vertex through a temporal path (also called {\em journey}~\cite{BFJ03}), that is, a path whose edges are used over a non-decreasing sequence of times. The times may also be required to be strictly increasing (strict journeys), both cases being carefully discussed in this paper. Temporal connectivity was considered in an early paper by Awerbuch and Even~\cite{AE84} (1984), and studied from a graph-theoretical point of view in the early  2000's in a number of seminal works including Kempe, Kleinberg, and Kumar~\cite{KKK02}, and Bui-Xuan, Ferreira, and Jarry~\cite{BFJ03} (see also~\cite{orda90} for an early study of graphs with time-dependent delays on the edges). More recently, it has been the subject of algorithmic studies, such as~\cite{AGMS17} and \cite{AF16} (discussed below), and \cite{restless,diluna,erlebach,matchings,viard,niedermeier}, which consider algorithms for computing structures related to temporal connectivity. Broad reviews of these topics can be found, \eg in~\cite{CFQS12,holme,michail}, although the list is non-exhaustive and the literature is rapidly evolving.

\subsection{Sparse Temporal Spanners and Related Work}




In a seminal paper, Kempe, Kleinberg, and Kumar~\cite{KKK02} asked whether every temporally connected graph on $n$ vertices has a sparse spanning subgraph that is also temporally connected. In their model, each edge has a single label indicating a time when it is present, thus the edges are identified by their labels, but the discussion is more general. What they are asking, essentially, is whether an analogue of spanning tree exists for temporal graphs when the labels are {\em already fixed}. We call such a temporally connected spanning subgraph a {\em temporal spanner}.

Kempe {\it et al.}\ answer their own question immediately (and negatively) for the particular case of $\Theta(n)$ density, by showing that hypercubes labelled in a certain way need all of their edges to achieve temporal connectivity, thus some temporal graphs of density $\Theta(n \log n)$ cannot be sparsified. The more general question, asking whether $o(n^2)$-sparse spanners always exist in dense temporal graphs remained open for more than a decade, and was eventually settled, again negatively, by Axiotis and Fotakis~\cite{AF16}. The proof in~\cite{AF16} exhibits an infinite family of temporally connected graphs with $\Theta(n^2)$ edges that do not admit $o(n^2)$-sparse spanners. Their construction can be adapted for strict and non-strict journeys.

On the positive side, Akrida {\it et al.}~\cite{AGMS17} show that, if the underlying graph $G$ is a complete graph and every edge is assigned a single globally-unique label, then it is always possible to find a temporal spanner of density ${n \choose 2} - \lfloor n/4 \rfloor$ edges (leaving however the asymptotic density unchanged). Akrida {\it et al.}~\cite{AGMS17} also prove that if the label of every edge in $G$ is chosen uniformly at random (from an appropriate interval), then almost surely the graph admits a temporal spanner with $O(n \log n)$ edges. Both~\cite{AF16} and~\cite{AGMS17} include further results related to the (in-)approximability of finding a {\em minimum} temporal spanner, which is out of the scope of this paper.  

By its nature, the problem of finding a temporal spanner in a temporal graph seems to be significantly different from its classical (\ie non-temporal) version, whether this version considers a static graph (see \eg~\cite{Chew86,NS07,PS89}) or the current topology of an updated dynamic graph (see \eg~\cite{arya,elkin,gottlieb}). The essential difference is that spanning trees always exist in standard (connected) graphs, thus one typically focuses on the trade-off between the density of a solution and a quality parameter like the {\em stretch factor}, rather than on the very existence of a sparse spanner.

\subsection{Contributions}
\label{sec:contributions}
In this paper, we establish that temporal graphs built on top of complete graphs {\em unconditionally} admit $O(n \log n)$-sparse temporal spanners when {\em non-strict} journeys are allowed. Furthermore, such a spanner can be computed in polynomial time. The case of strict journeys requires more discussion. Kempe {\it et al.} observed in~\cite{KKK02} that if every edge of a complete graph is given the same label, then this graph is temporally connected, but no multi-hop {\em strict} journey exists, thus none of the edges can be removed, and the problem is trivially unsolvable. To make the problem interesting when only strict journeys are allowed, one should constrain the edge labelling to avoid such a pathological situation. In the present case, we require that a sub-labelling of one label per edge exists in which any two adjacent edges have different labels. This formulation slightly generalizes the single-label global-unicity assumption made in~\cite{AGMS17} (although essentially equivalent) and eliminates the distinction between strict and non-strict journeys. Under this restriction, we establish that every temporal graph whose underlying graph is complete admits an $O(n \log n)$-sparse temporal spanner. Moreover, if the restricted labelling is given, then the spanner can be computed in polynomial time. (The problem of deciding whether a general labelling admits such a sub-labelling is not discussed here; it might be computationally hard.)

We start by observing that the above two settings one-way reduce to the setting where every edge has a single label and two adjacent edges have different labels. The reduction is ``one-way'' in the sense that the transformed instance may have fewer feasible journeys than the original instance, but all of these journeys correspond to valid journeys in the original instance, so a temporal spanner computed in the transformed instance is valid in the original instance. As a result, our main algorithm takes as input a complete graph $\G$ with single, locally distinct labels, and computes an $O(n\log n)$-sparse temporal spanner of $\G$ in polynomial time. 

In summary, we give the first positive answer to the question of whether sparse temporal spanners always exist in a class of dense graphs, focusing here on the case of complete graphs. This answer complements the negative answer by Axiotis and Fotakis~\cite{AF16} and motivates more investigation to understand where the transition occurs between their negative result (no sparse spanners exist in some dense temporal graphs) and our positive result (they essentially always exist in complete graphs). Our algorithms are based on a number of original techniques, which we think may be of more general interest for problems related to temporal connectivity.

The paper is organized as follows. In Section~\ref{sec:preliminaries}, we define the model and notation, and describe the one-way reductions that allow us to concentrate subsequently on single and distinct labels. We also mention two basic techniques of interest for this problem, although they are not used subsequently in the paper, namely the \textit{sub-clique} technique (from~\cite{AGMS17}) and the {\em pivoting} technique (introduced here). In Section~\ref{sec:delegation} we introduce the main concepts and techniques used in our algorithms, namely {\em delegation}, {\em dismountability}, and {\em $k$-hop dismountability}, whose purpose is to recursively self-reduce the problem to smaller graph instances. These techniques are subsequently combined into a more sophisticated algorithm that successfully computes a temporal spanner of $O(n\log n)$ edges. The first step, presented in Section~\ref{sec:fireworks}, is called {\em fireworks}. It is based on a generalization of delegation called {\em one-sided delegation} and results in a spanner of density (essentially) ${n \choose 2}/2$. Then, in Section~\ref{sec:iter_deleg}, we exploit a particular dichotomy in the structure of the residual instance, which allows us to sparsify the graph down to $O(n \log n)$ edges. In Section~\ref{sec:polynomial}, we review the main components of the algorithm, showing that its running time complexity is polynomial. Finally, we present a few open questions in Section~\ref{sec:tightness} and conclude with some remarks in Section~\ref{sec:conclusion}.

\section{Definitions and Basic Results}
\label{sec:preliminaries}

\subsection{Model and Definitions}
\label{sec:model}
Let $G=(V,E)$ be an undirected graph on $n=|V|$ vertices and $\lambda:E\to 2^{\mathbb{N}}$ a mapping that assigns to every edge of $E$ a non-empty set of integer labels. These labels can be seen as discrete times when the edge is present. In this paper, we refer to the resulting graph $\G=(G,\lambda)$ as a {\em temporal graph} (other models and terminologies exist, many of them being equivalent for the considered problem). If $\lambda$ is single-valued and locally injective (\ie adjacent edges have different labels), then we say that $\lambda$ is {\em simple}, and by extension, a temporal graph is simple if its labelling is simple.

A temporal path in $\G$ (also called {\em journey}), is a finite sequence of $k$ triplets ${\cal J}=\{(u_i, u_{i+1}, t_i)\}$ such that $(u_1,\dots,u_{k+1})$ is a path in $G$ and for all $1 \le i \leq k$, $\{u_i, u_{i+1}\} \in E$, $t_i \in \lambda(\{u_i, u_{i+1}\})$ and $t_{i+1} \ge t_{i}$. Strict temporal paths (strict journeys) are defined analogously by requiring that $t_{i+1} > t_{i}$. We say that a vertex $u$ can {\em reach} a vertex $v$ iff a journey exists from $u$ to $v$ (strict or non-strict, depending on the context). 
If every vertex can reach every other vertex, then $\G$ is {\em temporally connected}. Finally, observe that the distinction between strict and non-strict journeys does not exist in simple temporal graphs, as all the journeys are strict.

In general, one can define a {\em temporal spanner} of a temporally connected graph $\G=((V,E),\lambda)$ as a temporally connected spanning subgraph $\G'=((V',E'),\lambda')$ such that $V=V'$, $E' \subseteq E$ and $\lambda':E' \to 2^{\mathbb{N}}$ with $\lambda'(e) \subseteq \lambda(e)$ for all $e \in E'$.
Observe that, if $\G$ is simple, then the structure of a spanner of $\G$ is fully determined by the chosen subset of edges $E' \subseteq E$ and the temporal connectivity of the spanner is established by the labels on the edges in $E'$.
Thus, in such cases, we say that $E'$ itself {\em is} the spanner. Many of these notions are analogous to the ones considered in~\cite{AGMS17,AF16,KKK02}, although they are not referred to as ``spanners'' in these works. 

Finally, when the graph $G$ underlying a temporal graph $\G=(G,\lambda)$ is a complete graph, we call $\G$ a temporal clique. An example of a temporal spanner of a {\em simple temporal clique} is shown in Figure~\ref{fig:example}. The endpoints of the edge labelled 1 are temporally connected in the spanner by the journey with labels 4 and 6 in one direction and the journey with labels 3 and 7 in the other direction.

\begin{figure}[htb]

\centering
\newcommand{\hexafull}{
   \begin{tikzpicture}[scale=2]
     \path (0,0) coordinate (m);
     \path (m)+(90:1) node[defnode] (a){};
     \path (m)+(30:1) node[defnode] (b){};
     \path (m)+(330:1) node[defnode] (c){};
     \path (m)+(270:1) node[defnode] (d){};
     \path (m)+(210:1) node[defnode] (e){};
     \path (m)+(150:1) node[defnode] (f){};
     \tikzstyle{every node}=[circle,fill=white,inner
sep=1.1pt,font=\footnotesize]
     \tikzstyle{every path}=[shorten >= 5pt, shorten <= 5pt]
     \tikzstyle{chain}=[very thick, shorten <= 5pt, shorten >= 5pt]
     \draw (a) edge[chain] node[pos=.42]{5}(c);
     \draw (a) edge[chain] node[pos=.38]{3}(d);
     \draw (a) edge[chain] node[pos=.58]{6}(e);
     \draw (a) edge[chain] node[pos=.5]{4}(f);
     \draw (a) edge[chain] node[pos=.5]{7}(b);
     \draw (b) edge[chain] node[pos=.5]{1}(c);
     \draw (c) edge[chain] node[pos=.5]{2}(d);
     \draw (c) edge[chain] node[pos=.42]{3}(e);
     \draw (c) edge[chain] node[pos=.38]{6}(f);
     \draw (b) edge[chain] node[pos=.38]{8}(e);
     \draw (c) edge[chain] node[pos=.38]{7}(f);
     \draw (d) edge[chain] node[pos=.42]{10}(f);
     \draw [dashed, gray](b)-- node[pos=.58]{9}(f);
     \draw [dashed, gray](b)-- node[pos=.42]{4}(d);
     \draw [dashed, gray](d)-- node[pos=.5]{7}(e);
     \draw [dashed, gray](e)-- node[pos=.5]{1}(f);
   \end{tikzpicture}
}
\readlist\arg{10 7 3 8 6 0 5 12 13 2 4 11 9 14 1}
\hexafull
\caption{\label{fig:example} Example of a simple temporal clique and
  one of its temporal spanners (edges in bold). This spanner is not minimum (nor even minimal) and the reader may try to remove further edges.}
\end{figure}

\subsection{Generality of Simple Labellings}
We claimed in Section~\ref{sec:contributions} that if non-strict journeys are allowed, then one can transform a temporal clique $\G=(G,\lambda)$ with {\em unrestricted} labelling $\lambda$ into a clique $\HH=(G,\lambda_H)$ with {\em simple} labelling $\lambda_H$ such that any valid temporal spanner of $\HH$ induces a valid temporal spanner of $\G$. (As explained, the converse is false, but this is not a problem.) The reduction is in two steps. First, for every edge $e$, restrict $\lambda(e)$ to a single label chosen arbitrarily from $\lambda(e)$. Then in the resulting temporal clique, whenever $k$ edges incident to the same vertex have the same label $l$, these edges are relabelled with distinct values in the interval $[l, l+k-1]$ (arbitrarily) and the values of all the other labels in the graph that are larger than $l$ are increased by $k$. It is not difficult to see that if a journey exists in $\HH$, then the same sequence of edges induces a (possibly non-strict) journey in $\G$.

As for the second claim of the introduction, if strict journeys are the only ones allowed, then as explained, we require the existence of a {\em simple} sub-labelling $\lambda'$ of $\lambda$ (whose computation is not discussed). Here, it is even more direct that any journey based on the sub-labelling $\lambda'$ is available in the original instance.
Based on these arguments, the rest of the paper focuses on simple temporal cliques, and we sometimes drop the adjective ``simple'' when it is clear from the context.

\subsection{Basic Techniques}
\label{sec:preliminary}

We present here two basic sparsification techniques. The first (subcliques) is from previous works, and the second (pivotability) is original. Although these techniques are not directly used in the rest of the paper, they serve two main purposes. Firstly, both techniques are simple and can serve as a gentle warm-up for the reader. Secondly, the pivotability technique is a natural analog of Kosaraju's principle, which gives linear-size spanners in static directed graphs. However, we show that pivotability does not always work in temporal graphs, by presenting an infinite family of graphs that are not pivotable. This negative result thus motivates (and legitimates) the more sophisticated techniques presented in the rest of the paper.

\subsubsection{The Subclique Technique}
So far, the only existing approach for sparsifying simple temporal cliques is that of Akrida {\it et al.}~\cite{AGMS17}, who prove that one can always remove $\lfloor n/4 \rfloor$ edges without breaking temporal connectivity. Their approach is as follows. First, it is established that if $n=4$, then it is always possible to remove at least {\em one} edge. Then, as $n\to \infty$, one can arbitrarily partition the input clique into (essentially) $n/4$ subcliques of $4$ vertices each, and remove an edge from each subclique. The edges {\em between} subcliques are kept, thus the impact of each removal is limited to the corresponding subclique and the resulting graph is temporally connected. Before moving to other techniques, let us observe that this technique can be greatly improved as follows. 

\begin{observation}[Improving the subclique technique]
  \normalfont
  The type of partitioning used in~\cite{AGMS17} is {\em vertex-disjoint}, but it turns out that the same argument can be applied to {\em edge-disjoint} subcliques, with significant consequences. Indeed, by Wilson's Theorem~\cite{Wilson75}, the number of edge-disjoint cliques on $4$ vertices in a complete graph on $n$ vertices is $\lfloor n^2/12\rfloor$. (More generally, as $c\to 1$, graphs with minimum degree at least $cn$ have $\lfloor{n \choose 2} / {k \choose 2}\rfloor$ disjoint copies of $K_k$ for all $k$.) Therefore, one can remove $\lfloor n^2/12\rfloor=\Theta(n^2)$ edges.
\end{observation}

Although this technique allows us to remove $\Theta(n^2)$ edges, it seems unlikely that purely {\em structural} techniques like this one will lead to spanners of $o(n^2)$ edges. The techniques we develop in this paper are different in essence and consider the interplay of timestamps at a finer scale.

\subsubsection{The Pivotability Technique}

Another natural approach that one might think of is inspired by Kosaraju's principle for testing strong connectivity in a directed graph (see~\cite{AHU}). This principle relies on finding a vertex that all the other vertices can reach (through directed paths) and that can reach all these vertices in return. This condition is sufficient in standard graphs because paths are transitive. In the temporal setting, transitivity does not hold, but we can define a temporal analogue as follows. A \emph{pivot vertex} $p$ is a vertex such that all other vertices can reach $p$ by some time $t$ (through journeys) and $p$ can reach all other vertices {\em after} time $t$. The union of the tree of (incoming) journeys towards $p$ and the tree of (outgoing) journeys from $p$ forms a temporal spanner with at most $2(n-1)$ edges. Such a spanner is illustrated in Figure~\ref{fig:pivot}. 
Experiments suggest that pivot vertices exist asymptotically almost surely in random temporal cliques (where an instance corresponds to a random permutation of the edge labels $\{0, 1, 2, \ldots,{n \choose 2}-1\}$.
Unfortunately, arbitrarily large non-pivotable cliques exist, as shown next.

\begin{figure}[htb]
\centering
          \begin{tikzpicture}[scale=1.8]
    \path (0,0) coordinate (m);
    \tikzstyle{every node}=[circle,fill=white,inner sep=1.3pt]
    \path (m)+(90:1) node[defnode, fill=gray] (a){};    
    \draw (m)+(90:1.15) node[invisible]{$p$};
    \path (m)+(90-72:1) node[defnode] (b){};
    \path (m)+(90-144:1) node[defnode] (c){};
    \path (m)+(90+144:1) node[defnode] (d){};
    \path (m)+(90+72:1) node[defnode] (e){};
    
    \tikzstyle{every node}=[circle,fill=white,inner sep=1.3pt]
    \tikzstyle{every path}=[shorten >= 5pt, shorten <= 5pt]
    \tikzstyle{chain}=[very thick, shorten <= 5pt, shorten >= 5pt]
    \draw[->] (a) edge[chain, red] node[pos=.5]{5}(b);
    \draw [dashed, gray](a)-- node[pos=.5]{7}(c);
    \draw[<-] (a) edge[chain, green] node[pos=.5]{2}(d);
    \draw[<-] (a) edge[chain, green] node[pos=.5]{3}(e);
    \draw[->] (b) edge[chain, red] node[pos=.5]{8}(c);
    \draw[->] (b) edge[chain, green] node[pos=.5]{1}(d);
    \draw[->] (b) edge[chain, red] node[pos=.5]{6}(e);
    \draw[->] (c) edge[chain, red] node[pos=.5]{9}(d);
    \draw[->] (c) edge[chain, green] node[pos=.5]{0}(e);
    \draw [dashed, gray](d)-- node[pos=.5]{4}(e);

  \end{tikzpicture}
     \caption{\label{fig:pivot} Example of a pivotable graph. The (light) green edges belong to the tree of incoming journeys to pivot vertex $p$ (with $t=4$); the (darker) red edges belong to the tree of outgoing journeys; the dashed edges belong to neither.
}
 \end{figure}

\paragraph{Non-pivotable Graphs.}
 \label{apx:non-pivotable}

Here, we explain how to construct non-pivotable simple temporal cliques of arbitrary sizes. 
The construction ensures that there is a time $t$ before which no vertex can be reached by all vertices, and after which no vertex can reach all vertices. The choice of $t$ does not matter, as moving it forward or backward could only worsen one of the directions. Thus, the simple existence of such a $t$ rules out the existence of a pivot vertex.
The construction is first presented with respect to the $6$-vertex clique of Figure~\ref{fig:no_pivot_family}; then we explain how to generalize it. 
\begin{figure}[htb]
\centering
  \begin{tikzpicture}[scale=1.8]
    \path (0,0) coordinate (m);
    \tikzstyle{every node}=[circle,fill=white,inner sep=1.3pt]
    \draw (m)+(90:1.2) node[invisible]{$u$};
    \path (m)+(90:1) node[defnode] (a){};
    \draw (m)+(30:1.2) node[invisible]{$v$};
    \path (m)+(30:1) node[defnode] (b){};
    \draw (m)+(-30:1.2) node[invisible]{$w$};
    \path (m)+(330:1) node[defnode] (c){};
    \path (m)+(270:1) node[defnode] (d){};
    \path (m)+(210:1) node[defnode] (e){};
    \path (m)+(150:1) node[defnode] (f){};
    \tikzstyle{every node}=[circle,font=\footnotesize,fill=white,inner sep=.8pt]
    \tikzstyle{every path}=[shorten >= 5pt, shorten <= 5pt]
    \tikzstyle{chain}=[ultra thick, shorten <= 5pt, shorten >= 5pt]
    \draw (a) edge[chain] node[pos=.5]{0}(b);
    \draw (a) edge[dashed, gray] node[pos=.42]{}(c);
    \draw (a) edge[chain] node[pos=.38]{5}(d);
    \draw (a)edge[chain] node[pos=.58]{6}(e);
    \draw (a)edge[chain] node[pos=.5]{2}(f);
    \draw (b) edge[chain] node[pos=.5]{1}(c);
    \draw (b)edge[dashed, gray] node[pos=.42]{}(d);
    \draw (b)edge[dashed, gray] node[pos=.38]{}(e);
    \draw (b)edge[dashed, gray] node[pos=.58]{}(f);
    \draw (c) edge[dashed, gray] node[pos=.5]{}(d);
    \draw (c)edge[dashed, gray] node[pos=.42]{}(e);
    \draw (c)edge[dashed, gray] node[pos=.38]{}(f);
    \draw (d)edge[chain] node[pos=.5]{4}(e);
    \draw (d)edge[chain] node[pos=.42]{7}(f);
    \draw (e) edge[chain] node[pos=.5]{3}(f);

  \end{tikzpicture}
  \qquad
  \begin{tikzpicture}[scale=1.8]
    \path (0,0) coordinate (m);
    \tikzstyle{every node}=[circle,fill=white,inner sep=1.3pt]
    \draw (m)+(90:1.2) node[invisible]{$u$};
    \path (m)+(90:1) node[defnode] (a){};
    \draw (m)+(30:1.2) node[invisible]{$v$};
    \path (m)+(30:1) node[defnode] (b){};
    \draw (m)+(-30:1.2) node[invisible]{$w$};
    \path (m)+(330:1) node[defnode] (c){};
    \path (m)+(270:1) node[defnode] (d){};
    \path (m)+(210:1) node[defnode] (e){};
    \path (m)+(150:1) node[defnode] (f){};
    \tikzstyle{every node}=[circle,font=\footnotesize,fill=white,inner sep=.8pt]
    \tikzstyle{every path}=[shorten >= 5pt, shorten <= 5pt]
    \tikzstyle{chain}=[ultra thick, shorten <= 5pt, shorten >= 5pt]
    \draw (a) edge[dashed, gray](b);
    \draw (a) edge[chain] node[pos=.42]{11}(c);
    \draw (a) edge[dashed, gray] (d);
    \draw (a)edge[dashed, gray](e);
    \draw (a)edge[dashed, gray](f);
    \draw (b) edge[dashed, gray](c);
    \draw (b)edge[chain] node[pos=.42]{8}(d);
    \draw (b)edge[chain] node[pos=.38]{9}(e);
    \draw (b)edge[chain] node[pos=.58]{10}(f);
    \draw (c) edge[chain] node[pos=.5]{14}(d);
    \draw (c)edge[chain] node[pos=.42]{13}(e);
    \draw (c)edge[chain] node[pos=.38]{12}(f);
    \draw (d)edge[dashed, gray] (e);
    \draw (d)edge[dashed, gray](f);
    \draw (e) edge[dashed, gray](f);
  \end{tikzpicture}
  
\caption{\label{fig:no_pivot_family}Example of a non-pivotable clique seen as the union of two specific subgraphs which represent the periods $[0,7]$ (left) and $[8,14]$ (right).}
\end{figure}
Thus, let $n=6$. In this case, let $t=7$ and let us consider the two periods $[0,7]$ and $[8,14]$. Looking at the graph, observe that none of the vertices can be reached by all the others during the first period. This is true because (1) the only vertex that $w$ can reach is $v$, making $v$ the only candidate, and (2) none of the other vertices (except $u$) can reach~$v$. Similarly, none of the vertices can reach all the others in the second period. This is true because (1) $u$ can only be reached by $w$, making $w$ the only candidate, and (2) $w$ cannot reach $v$ in the second period. 

The construction can be generalized to any larger value of $n$ by choosing three vertices to play the same role as $u$, $v$, and $w$. The graph of the first period corresponds to a subclique on $n-2$ vertices (including $u$, but not $v$ and $w$), plus the two edges $\{u,v\}$ and $\{v, w\}$. Thus $t={n-2 \choose 2} + 1$. The labelling assigns $0$ to $\{u,v\}$, $1$ to $\{v,w\}$, and all values in $[2,t]$ to the edges of the subclique (arbitrarily). By the same argument as above, none of the vertices can be reached by all others during the first period. As for the second period, the labelling must ensure that all the edges of $w$ have a larger label than all the edges of $v$, and that $\{u,w\}$ is assigned the smallest label among the edges incident to $w$, which results in direct applicability of the same argument as above. Thus, none of the vertices can reach all the others during the second period.


\section{Delegation and Dismountability}
\label{sec:delegation}

This section introduces a number of basic techniques which are subsequently used (and extended) in Sections~\ref{sec:fireworks} and~\ref{sec:iter_deleg}.
Given a vertex $v$, we use $e^-(v)$ to denote the edge with smallest label incident to $v$, and $e^+(v)$ analogously for the largest label.

\begin{lemma}
  \label{fact:minus}
  Given a temporal clique $\G$, if $\{u,v\}=e^-(v)$, then $u$ can reach all other vertices through $v$. Symmetrically, if $\{u,w\}=e^+(w)$, then all vertices can reach $u$ through $w$.
\end{lemma}

\begin{proof}
  If $\{u,v\}=e^-(v)$, then $v$ has an edge with every other vertex after that time, thus a journey exists from $u$ to every vertex through $v$. A symmetrical argument applies in the second case.
\end{proof}

Observe that Lemma~\ref{fact:minus} holds only when the underlying graph is complete. This property makes it possible for a vertex $u$ to {\em delegate} its emissions to a vertex $v$, \ie exploit the fact that $v$ can still reach all the other vertices {\em after} interacting with $u$, thus none of $u$'s other edges are required for reaching the other vertices. By a symmetrical argument, $u$ can delegate its receptions to a vertex $w$ if $w$ can be reached by all the other vertices {\em before} interacting with $u$, so the other edges of $u$ are not needed for being reached by other vertices.

\subsection{Dismountability}
The delegation concept suggests an interesting technique to construct temporal spanners. We say that a vertex $u$ in a temporal clique $\G$ is \emph{dismountable} if there exist two other vertices $v$ and $w$ such that $\{u, v\} = e^-(v)$ and $\{u, w\} = e^+(w)$, i.e., $u$ can delegate both its emissions and its receptions. The existence of such a vertex enables a self-reduction of the spanner construction as follows: 
select $e^-(v)$ and $e^+(w)$ for future inclusion in the spanner, then recurse on the smaller clique $\G[V\setminus u]$, as illustrated in Figure~\ref{fig:full_dismountability}. More precisely:

\begin{theorem}[Dismountability]
  \label{thm:dismountability}
  Let $\G$ be a temporal clique, and let $u, v, w$ be three vertices in $\G$ such that $\{u, v\} = e^-(v)$ and $\{u, w\} = e^+(w)$. Let $S'$ be a temporal spanner of $\G[V\setminus u]$. Then $S=S'\cup \{\{u, v\}, \{u, w\}\}$ is a temporal spanner of $\G$.
\end{theorem}

\begin{proof}
  Since $\{u,v\}=e^-(v)$, all edges incident to $v$ in $S'$ have a larger label than $\{u,v\}$, thus $u$ can reach all the vertices through $v$ using only the edges of $S'$. A symmetrical argument implies that all vertices in $\G$ can reach $u$ through $w$ using only $\{u,w\}$ and the edges of $S'$.
\end{proof}

We call a graph \emph{dismountable} if it contains a dismountable vertex. It is said to be {\em fully dismountable} if one can find an ordering of $V$ that allows for a recursive dismounting of the graph until the residual instance has two vertices. 
An example of fully dismountable graph is given in Figure~\ref{fig:full_dismountability}.

\begin{figure}[htb]
\centering
     \begin{subfigure}[b]{0.22\textwidth}
          \centering
          \resizebox{\linewidth}{!}{
          \centering
          \begin{tikzpicture}[scale=2.5]
    \path (0,0) coordinate (m);
    \tikzstyle{every node}=[circle,fill=white,inner sep=1.3pt]
    \path (m)+(90:1) node[defnode, fill=gray] (a){};
    \path (m)+(90-72:1) node[defnode] (b){};
    \path (m)+(90-144:1) node[defnode] (c){};
    \path (m)+(90+144:1) node[defnode] (d){};
    \path (m)+(90+72:1) node[defnode] (e){};
    
    \tikzstyle{every node}=[circle,fill=white,inner sep=1.3pt,font=\LARGE]
    \tikzstyle{every path}=[thin,gray,shorten >= 5pt, shorten <= 5pt]
    \tikzstyle{chain}=[ultra thick, black,shorten <= 5pt, shorten >= 5pt]
    \draw (a) edge[chain] node[pos=.5]{9}(e);
    \draw (a)-- node[pos=.5]{0}(b);
    \draw (a) edge[chain] node[pos=.5]{3}(c);
    \draw (a)-- node[pos=.5]{2}(d);
    \draw (b) --node[pos=.5]{4}(c);
    \draw (b) -- node[pos=.5]{6}(d);
    \draw (b) -- node[pos=.5]{1}(e);
    \draw (c) -- node[pos=.5]{5}(d);
    \draw (c)-- node[pos=.5]{7}(e);
    \draw (d) --node[pos=.5]{8}(e);

  \end{tikzpicture}
          }  
     \end{subfigure}
     \begin{subfigure}[b]{0.22\textwidth}
          \centering
          \resizebox{\linewidth}{!}{
          \centering
          \begin{tikzpicture}[scale=2.5]
    \path (0,0) coordinate (m);
    \tikzstyle{every node}=[circle,fill=white,inner sep=1.3pt]
    \path (m)+(90-72:1) node[defnode] (b){};
    \path (m)+(90-144:1) node[defnode] (c){};
    \path (m)+(90+144:1) node[defnode] (d){};
    \path (m)+(90+72:1) node[defnode, fill=gray] (e){};
    
    \tikzstyle{every node}=[circle,fill=white,inner sep=1.3pt,font=\LARGE]
    \tikzstyle{every path}=[thin,gray,shorten >= 5pt, shorten <= 5pt]
    \tikzstyle{chain}=[ultra thick, black,shorten <= 5pt, shorten >= 5pt]
    \draw (b) --node[pos=.5]{4}(c);
    \draw (b) -- node[pos=.5]{6}(d);
    \draw (b) edge[chain] node[pos=.5]{1}(e);
    \draw (c) -- node[pos=.5]{5}(d);
    \draw (c) edge[chain] node[pos=.5]{7}(e);
    \draw (d) --node[pos=.5]{8}(e);

  \end{tikzpicture}
          }  
     \end{subfigure}
     \begin{subfigure}[b]{0.2\textwidth}
          \resizebox{\linewidth}{!}{
          \begin{tikzpicture}[scale=3]
    \path (0,0) coordinate (m);
    \tikzstyle{every node}=[circle,fill=white,inner sep=1.3pt]
    \path (m)+(90-72:1) node[defnode, fill=gray] (b){};
    \path (m)+(90-144:1) node[defnode] (c){};
    \path (m)+(90+144:1) node[defnode] (d){};
    
    \tikzstyle{every node}=[circle,fill=white,inner sep=1.3pt,font=\LARGE]
    \tikzstyle{every path}=[thin,gray,shorten >= 5pt, shorten <= 5pt]
    \tikzstyle{chain}=[ultra thick, black,shorten <= 5pt, shorten >= 5pt]
    \draw (b) edge[chain] node[pos=.5]{4}(c);
    \draw (b) edge[chain] node[pos=.5]{6}(d);
    \draw (c) edge node[pos=.5]{5}(d);

  \end{tikzpicture}
          }  
          \label{fig:K3}
     \end{subfigure}
     \hspace{20pt}
     \begin{subfigure}[b]{0.22\textwidth}
          \resizebox{\linewidth}{!}{
          \begin{tikzpicture}[scale=2.5]
    \path (0,0) coordinate (m);
    \tikzstyle{every node}=[circle,fill=white,inner sep=1.3pt]
    \path (m)+(90:1) node[defnode] (a){};
    \path (m)+(90-72:1) node[defnode] (b){};
    \path (m)+(90-144:1) node[defnode] (c){};
    \path (m)+(90+144:1) node[defnode] (d){};
    \path (m)+(90+72:1) node[defnode] (e){};
    
    \tikzstyle{every node}=[circle,fill=white,inner sep=1.3pt,font=\LARGE]
    \tikzstyle{every path}=[shorten >= 5pt, shorten <= 5pt]
    \tikzstyle{chain}=[ultra thick, shorten <= 5pt, shorten >= 5pt]
    \draw [dashed, gray](a) -- node[pos=.5]{0}(b);
    \draw (a)edge[chain] node[pos=.5]{3}(c);
    \draw [dashed, gray](a)-- node[pos=.5]{2}(d);
    \draw (a) edge[chain] node[pos=.5]{9}(e);
    \draw (b) edge[chain] node[pos=.5]{4}(c);
    \draw (b) edge[chain] node[pos=.5]{6}(d);
    \draw (b) edge[chain] node[pos=.5]{1}(e);
    \draw (c) edge[chain] node[pos=.5]{5}(d);
    \draw (c) edge[chain] node[pos=.5]{7}(e);
    \draw [dashed, gray](d) -- node[pos=.5]{8}(e);

  \end{tikzpicture}
          }  
     \end{subfigure}
     \caption{\label{fig:full_dismountability} Example of a fully dismountable graph and the resulting spanner.}
 \end{figure}

\begin{lemma}[Spanners based on dismountability]
  \label{fact:dismounted-spanner}
  If a graph can be fully dismounted, then the resulting spanner will have $2(n-2)+1=2n - 3$ edges.
\end{lemma}

\begin{proof}
Each dismounted vertex contributes two edges to the spanner and the residual instance contributes one edge.
\end{proof}

Unfortunately, there are arbitrarily large temporal cliques which are not fully dismountable as shown in Subsection~\ref{apx:non-dismountable}.

\subsection{$k$-hop Delegation and $k$-hop Dismountability}
\label{subsec:k-hop}
The dismountability technique can be generalized to multi-hop journeys. Let $J$ be a journey from vertex $u$ to vertex $v$ through vertices $u=u_0, u_1, \ldots, u_k=v$ with $\{u_{k-1},u_k\}=e^-(v)$. The key observation is that $u$ can delegate its emissions to $v$ even if $\{u_{i-1},u_i\} \ne e^-(u_i)$ for some $i$.
Indeed, it is sufficient that the {\em last} edge of a journey from $u$ to $v$ is the minimum (at $v$) in order to delegate $u$'s emissions to $v$. Symmetrically, it is sufficient that the {\em first} edge of a journey from $w$ to $u$ is $e^+(w)$ in order to delegate $u$'s receptions to $w$.
Accordingly, a vertex $u$ is called {\em $k$-hop dismountable} if one can find two other vertices $v$ and $w$ (possibly identical if $k> 1$) such that there are journeys of {\em at most} $k$ hops (1) from $u$ to $v$ that arrives at $v$ through $e^-(v)$, and (2) from $w$ to $u$ that leaves $w$ through $e^+(w)$. See Figure~\ref{fig:dismountability}(b) for an illustration.

We call a graph {\em $k$-hop dismountable} if it contains a $k$-hop dismountable vertex. It is said to be {\em fully $k$-hop dismountable} if one can find an ordering of $V$ that allows for a recursive $k$-hop dismounting of the graph until the residual instance has two vertices. Note that $k$-hop dismountability implies $k'$-hop dismountability for $k' > k$ by definition but the converse is not true.

\begin{figure}[htb]
  \centering
\begin{subfigure}[b]{0.30\textwidth}
  \centering
\begin{tikzpicture}[scale=2.3]
    \path (0,0) coordinate (m);
    \draw (0,0) ellipse (1cm and .5cm) node[invisible,anchor=south] {};
    \tikzstyle{every node}=[circle,fill=white,inner sep=1.3pt]
    \path (m)+(-.5,0) node[defnode] (v) {};
    \path (m)+(.5,0) node[defnode] (w) {};
    \path (m)+(0,-.7) node[defnode] (u){};
    \draw (v)+(30:.15) node[invisible,anchor=north] {$v$};
    \draw (u)+(30:.15) node[invisible,anchor=north] {$u$};
    \draw (w)+(30:.15) node[invisible,anchor=north] {$w$};
    \tikzstyle{every node}=[fill=white,inner sep=2pt]
    \tikzstyle{every path}=[shorten >= 1pt, shorten <= 1pt]

    \draw (u)-- node[pos=.7]{\small $e^-(v)$}(v);
    \draw (u)-- node[pos=.7]{\small $e^+(w)$}(w);
  \end{tikzpicture}
\caption{Dismountability principle.}
\end{subfigure}
\begin{subfigure}[b]{0.35\textwidth}
  \centering
\begin{tikzpicture}[scale=1.7]
    \path (0,0) coordinate (m);
    \draw (0,0) ellipse (1.25cm and .75cm) node[invisible,anchor=south] {};
    \tikzstyle{every node}=[circle,fill=white,inner sep=1.3pt]
    \path (m)+(60:.5) node[defnode] (w) {};
    \path (m)+(-.6,-.2) node[defnode] (v) {};
    \path (m)+(-15:.65) node[defnode] (passby) {};
    \path (m)+(270:1) node[defnode] (u){};
    \draw (v) node[right=3pt,invisible] {$v$};
    \draw (u) node[right=3pt,invisible] {$u$};
    \draw (w) node[right=3pt,invisible] {$w$};
    \tikzstyle{every node}=[circle,fill=white,inner sep=0pt]
    \tikzstyle{every path}=[shorten >= 1pt, shorten <= 1pt]
    
    \draw (w)-- node[rectangle,inner sep=.5pt,pos=.42]{\footnotesize $e^+(w)$}(passby);
    \draw (passby)-- node[rectangle,inner sep=.5pt,pos=.35]{\footnotesize $>$$e^+(w)$}(u);
    \draw (v)-- node[rectangle,inner sep=0.5pt,pos=.35]{\footnotesize $e^-(v)$}(u);
  \end{tikzpicture}
\caption{Example of $2$-hop dismountability.}
\end{subfigure}
\caption{\label{fig:dismountability}Illustration of the principle of dismountability and $k$-hop dismountability.}
\end{figure}

Temporal spanners can be obtained in a similar way to $1$-hop dismountability by selecting all of the edges involved in these journeys for inclusion in the spanner. However, only the edges adjacent to the dismounted vertex are removed in the recursion, thus some edges used in a multi-hop journey may be selected several times over the recursion (to our advantage). We can then state an analogous fact for $k$-hop dismountability as follows. 

\begin{lemma}
  \label{fact:k-hop-number}
If a temporal graph $\mathcal{G}$ is fully $k$-hop dismountable, then this process yields a temporal spanner with at most $2k(n-2)+1< 2kn$ edges.
\end{lemma}

\begin{proof}
Each dismounted vertex contributes at most $2k$ edges to the spanner and the residual instance contributes one edge.
\end{proof}

Unfortunately, again there exist arbitrarily large graphs (in fact, even cliques) which are not $k$-hop dismountable (for any $k$), as discussed next.
Nonetheless, $k$-hop dismountability is a core component in the more sophisticated techniques presented in this paper.

\subsection{Non-dismountable Cliques}
\label{apx:non-dismountable}

Here, we explain how to construct an infinite family of simple temporal cliques which are not $k$-hop dismountable for any $k$. Such a clique can be constructed for any even $n \ge 4$.
We construct this family by first partitioning the vertices into two equal parts $V^-$ and $V^+$, and then adding the edges (within and between these parts) using time labels in the set $\{0, 1, 2, ..., {n \choose 2}-1\}$, as follows (see \cref{fig:undism_family}).
First, create a subclique among all vertices of $V^-$ and put the ${n/2} \choose 2$ smallest labels on these edges. Similarly, create a subclique among all vertices of $V^+$ and put the ${n/2} \choose 2$ largest labels on the edges. Then, create two edge-disjoint maximum matchings, denoted $M^-$ and $M^+$, between the two subcliques (so each matching has $\frac{n}{2}$ edges). On the edges of $M^-$, put the $n/2$ smallest remaining labels, and on $M^+$ put the $n/2$ largest remaining labels. 
Add all of the missing edges between $V^-$ and $V^+$ (call this set of edges $E'$) to obtain a clique on $n$ vertices, and label them with the remaining labels.
Thus, in increasing order of labels, we have the edges among the vertices of $V^-$, edges of $M^-$, edges of $E'$, edges of $M^+$, and finally the edges among the vertices of $V^+$. 

\begin{figure}
    \centering
    \includegraphics[scale=.55]{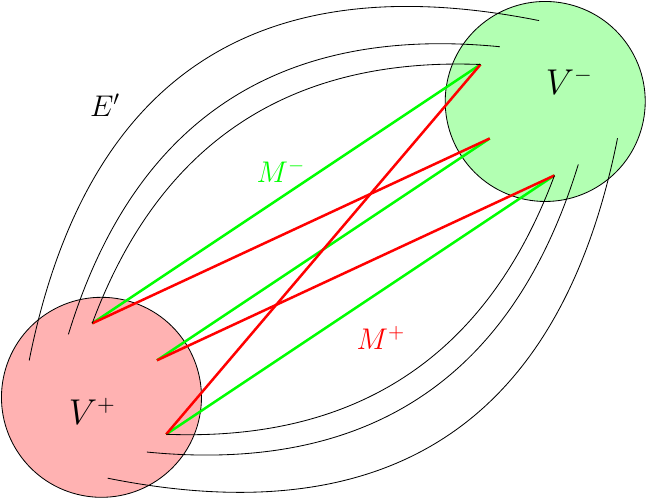}
    \caption{Construction of a non-dismountable graph.}
    \label{fig:undism_family}
\end{figure}

Clearly, no vertex $u \in V^-$ can reach a vertex $v$ using $e^-(v)$, and no vertex $u \in V^+$ can be reached by a vertex $w$ using $e^+(w)$, regardless of how many hops are allowed. Thus, this construction creates a ($k$-hop) non-dismountable temporal graph.

\begin{remark}
The fireworks technique (presented in \cref{sec:fireworks}) always finds a $2$-hop dismountable node if $n$ is odd. Thus the restriction that $n$ is even in the above construction is necessary (at least for $k \geq 2$).
\end{remark}

Finally, observe that the non-dismountable cliques obtained by this construction are pivotable. However, some graphs (and even some cliques) exist which are neither pivotable nor dismountable, as illustrated in Figure~\ref{fig:neither}. Experiments that we have conducted suggest that there exist arbitrarily large graphs which are non-pivotable and non-dismountable, however, we leave the characterization of an infinite family with these properties as an open question.

\begin{figure}[h]
  \centering
  \begin{tikzpicture}[scale=2]                                                
    \path (0,0) node[defnode] (d){};                                            
    \path (0,1) node[defnode] (a){};                                            
    \path (1,0) node[defnode] (c){};                                            
    \path (1,1) node[defnode] (b){};                                            
    \draw (a)-- node[above]{1}(b);                                            
    \draw (a)-- node[left,pos=.4]{4}(c);                                      
    \draw (d)-- node[left]{2}(a);                                             
    \draw (b)-- node[right]{3}(c);                                            
    \draw (b)-- node[right,pos=.4]{5}(d);                                     
    \draw (c)-- node[below]{6}(d);                                         
  \end{tikzpicture}         
  \caption{\label{fig:neither}Example of a temporal clique that is neither dismountable nor pivotable.}
\end{figure}

\section{The Fireworks Technique}
\label{sec:fireworks}

In this section, we present an algorithm called {\em fireworks}, which exploits a system of multi-hop delegations among vertices. In particular, we take advantage of {\em one-sided delegations}, in which a vertex delegates only its emissions, or only its receptions. The combination of many such delegations is shown to lead to the removal of essentially half of the edges of the input clique. The residual instance has a particular structure that is exploited in Section~\ref{sec:iter_deleg} to obtain the final $O(n \log n)$-sparse spanners.

\subsection{Forward Fireworks}
\label{subsec:forward_fw}

The purpose of fireworks is to mutualize several one-sided delegations in a transitive way, so that many vertices do not need to reach the other vertices directly. Given a temporal clique $\G=(G,\lambda)$ with $G=(V,E)$, define the directed graph ${G^-}=(V,E^-)$ such that $(u,v)\in E^-$ iff $\{u,v\} = e^-(v)$, except that, if $e^-(u) = e^-(v)$ for some $u$ and $v$, only one of the arcs is included (chosen arbitrarily).

\newcommand{\EM}{{E^-}}
\newcommand{\EMP}{{E^-}'}

\begin{lemma}
    \label{lem:directed-paths-minus}
	The sequences of labelled edges in $\G$ corresponding to directed paths in ${G^-}$ are journeys in $\G$.
\end{lemma}

\begin{proof}[Proof (by contradiction).]
Let $(u_0, u_1), (u_1, u_2), ..., (u_{k-1}, u_{k})$ be a directed path in ${G^-}$ and suppose that the corresponding path in $\G$ is \emph{not} a journey. Then it must be the case that the label of an edge $(u_{i-1}, u_i)$ is greater than the label on the adjacent edge $(u_i, u_{i+1})$ for some $i$. Then $\{u_{i-1},u_i\} \ne e^-(u_i)$ which is impossible.
\end{proof}

By construction, $E^-$ induces a disjoint set of {\em out-trees} (one source, possibly several sinks). We transform $G^-$ into a disjoint set ${\cal T}^-= (V,E^-_T)$ of {\em in-trees} (one sink, possibly several sources) as follows (see also Figure~\ref{fig:minchain} for an illustration). 
Let $E^-_T$ be initialized as a copy of $E^-$. 
For every $v$ with outdegree at least $2$ in $G^-$, let $(v,u_1), ..., (v, u_{\ell})$ be its out-arcs with $(v, u_{\ell})$ being the one with the {\em largest} label.
For every $i < \ell$, if $u_i$ is a sink vertex, then flip the direction of $(v,u_i)$ in $E^-_T$ (\ie replace $(v,u_i)$ by $(u_i,v)$ in $E^-_T$); otherwise remove $(v,u_i)$ from $E^-_T$. Let ${\cal T}^-= (V,E^-_T)$ be the resulting set of in-trees ${\cal T}_1^-, ..., {\cal T}_k^-$ (containing possibly more in-trees than the number of initial out-trees). 

\begin{figure}[htb]
  \pgfmathsetmacro{\ex}{0}
  \pgfmathsetmacro{\ey}{1}
  \centering
\begin{tikzpicture}[scale=1]
  \tikzstyle{every node}=[defnode]
  \path (3.98,7.68) node [] (v0) {};
  \path (3.92,8.5) node [] (v1) {};
  \path (3.1,8.64) node [emitter] (v2) {};
  \path (3.5,9.5) node [] (v3) {};
  \path (4.28,9.4) node [] (v4) {};
  \path (4.88,8.76) node [emitter] (v5) {};
  \path (4.72,10.24) node [emitter] (v6) {};
  \path (3.76,10.5) node [emitter] (v7) {};
  \path (3.12,10.44) node [emitter] (v8) {};
  \path (2.72,9.84) node [emitter] (v9) {};
  \path (3.3,7.9) node [] (v10) {};
  \path (2.6,8.1) node [emitter] (v11) {};
  \tikzstyle{every node}=[]
  \tikzstyle{every path}=[thick, green,shorten <=2pt,shorten >=2pt];
  \draw [->, very thick, OliveGreen] (v0)--(v1);
  \draw [->] (v1)--(v2);
  \draw [->] (v1)--(v3);
  \draw [->] (v1)--(v4);
  \draw [->, very thick, OliveGreen] (v1)--(v5);
  \draw [->, very thick, OliveGreen] (v4)--(v6);
  \draw [->, very thick, OliveGreen] (v3)--(v7);
  \draw [->] (v3)--(v8);
  \draw [->] (v3)--(v9);
  \draw [->, very thick, OliveGreen] (v10)--(v11);
  
\end{tikzpicture}
~~$\to$~~
\begin{tikzpicture}[scale=1]
  \tikzstyle{every node}=[defnode]
  \path (3.98,7.68) node [] (v0) {};
  \path (3.92,8.5) node [] (v1) {};
  \path (3.1,8.64) node [] (v2) {};
  \path (3.5,9.5) node [] (v3) {};
  \path (4.28,9.4) node [] (v4) {};
  \path (4.88,8.76) node [emitter] (v5) {};
  \path (4.72,10.24) node [emitter] (v6) {};
  \path (3.76,10.5) node [emitter] (v7) {};
  \path (3.12,10.44) node [] (v8) {};
  \path (2.72,9.84) node [] (v9) {};
  \path (3.3,7.9) node [] (v10) {};
  \path (2.6,8.1) node [emitter] (v11) {};
  \tikzstyle{every path}=[thick,green,shorten <=2pt,shorten >=2pt];
  \draw [->, very thick, OliveGreen] (v0)--(v1);
  \draw [<-] (v1)--(v2);
  \draw [->, very thick, OliveGreen] (v1)--(v5);
  \draw [->, very thick, OliveGreen] (v4)--(v6);
  \draw [->, very thick, OliveGreen] (v3)--(v7);
  \draw [<-] (v3)--(v8);
  \draw [<-] (v3)--(v9);
  \draw [->, very thick, OliveGreen] (v10)--(v11);
\end{tikzpicture}

\caption{\label{fig:minchain} Example of transformation from a disjoint set of out-trees $(V,E^-)$ to a disjoint set of in-trees $(V,E^-_T)$. The coloured vertices represent sink vertices, and the darker arrows represent out-arcs with the largest label (for the vertex on the tail side).}
\end{figure}

\begin{lemma}
  \label{fact:Tminus}
The set of in-trees ${\cal T}^-= (V,E^-_T)$ has the following properties:
\begin{enumerate}
\item Directed paths in ${\cal T}^-$ correspond to journeys in $\G$.
\item Every vertex belongs to exactly one tree.
\item Every tree contains at least two vertices.
\item There is a unique sink in each tree.
\item The unique arc incident to a sink $s$ corresponds to $e^-(s)$.
\end{enumerate}
\end{lemma}

\begin{proof}
Each property is proved as follows.
\begin{enumerate}
\item This follows from Lemma~\ref{lem:directed-paths-minus} because an arc $(v,u_i)$ is only replaced by $(u_i,v)$ if the label of $(v,u_i)$ is less than the label of another arc $(v,u_{\ell})$, so $(u_i,v),(v,u_{\ell})$ is a journey in $\G$.
\item The transformation from out-trees to in-trees does not add arcs between trees.
\item Every out-tree contains at least one arc (the arc of its source with minimum label). An out-tree with a single arc is not changed by the transformation. Arcs to sinks are not removed (but they may be flipped). A vertex $v$ that loses arcs during the transformation retains at least its arc with minimum label. In summary, no vertex loses all of its arcs.
\item Each vertex in an in-tree has at most one out-arc.
\item The unique arc to a sink is an arc that was not flipped by the transformation.
\end{enumerate}\vspace{-20pt}
\end{proof}

Observe that some of the journeys induced by the arcs of ${\cal T}^-$ may include intermediate hops where the arc's label is not locally minimum for its head endpoint. However, as already discussed in Section~\ref{subsec:k-hop}, a delegation only requires that the label of the last hop of a journey be locally minimum, and that is the case here (Lemma~\ref{fact:Tminus}.5).

The motivation behind this construction is that all the vertices in each in-tree are able to delegate their emissions to the corresponding sink vertex. For this reason, the sink vertex will be called an {\em emitter} in the rest of the paper.
An important consequence of our construction is that the number of emitters in ${\cal T}^-$ cannot exceed half of the total number of vertices.
\begin{lemma}
\label{lemma:half_n_emitters}
The number of emitters in ${\cal T}^-$ is at most $n/2$.
\end{lemma}
\begin{proof}
After the transformation from $E^-$ to $E^-_T$, there is only one emitter in each in-tree ${\cal T}_i^- \in {\cal T}^-$ (Lemma~\ref{fact:Tminus}.4), and there are at most $n/2$ trees because each one contains at least $2$ vertices (Lemma~\ref{fact:Tminus}.3).
\end{proof}

We are now ready to define a temporal spanner based on ${\cal T}^-$, which consists of the union of all edges corresponding to arcs in the in-trees and all edges incident to at least one emitter. More formally, let $S_T^-=\{\{u,v\} \in E :(u,v) \in {\cal T}^-\} \cup \{\{u,v\} \in E: u$ is an emitter$\}$.

\begin{theorem}
  \label{th:forward-valid}
  $S_T^-$ is a temporal spanner of $\G$.
\end{theorem}
\begin{proof}
By Lemma~\ref{fact:Tminus},
every vertex $v$ of $\G$ that is a non-emitter in ${\cal T}^-$ can reach an emitter $s$ through an edge $e^-(s)$. Furthermore, the inclusion of all edges incident to a vertex $s$ that is an emitter in ${\cal T}^-$ ensures that $v$ can still reach all other vertices afterwards and so can $s$.
Therefore, every vertex can reach all other vertices by using only edges from $S_T^-$.
\end{proof}

We call this type of spanner a {\em forward fireworks spanner}. An example is given in Figure~\ref{fig:forward-fireworks-spanner}, the corresponding journeys being depicted on the right side. Note that the orientations of the arcs with labels 0 and 1 in this example can be chosen arbitrarily when constructing $E^-$. If the orientations are chosen as shown in the right part of the figure, then $E^- = E^-_T$.

\begin{figure}[htb]
\centering
\newcommand{\hexafull}{
  \begin{tikzpicture}[scale=1.8]
    \path (0,0) coordinate (m);
    \path (m)+(90:1) node[defnode] (a){};
    \path (m)+(30:1) node[defnode] (b){};
    \path (m)+(330:1) node[defnode] (c){};
    \path (m)+(270:1) node[defnode] (d){};
    \path (m)+(210:1) node[defnode] (e){};
    \path (m)+(150:1) node[defnode] (f){};
    \tikzstyle{every node}=[circle,fill=white,inner sep=.8pt,font=\footnotesize]
    \tikzstyle{every path}=[shorten >= 5pt, shorten <= 5pt]
    \tikzstyle{chain}=[very thick, shorten <= 5pt, shorten >= 5pt]
    \draw (a) edge[chain] node[pos=.42]{7}(c);
    \draw (a) edge[chain] node[pos=.38]{3}(d);
    \draw (a) edge[chain] node[pos=.58]{8}(e);
    \draw (a) edge[chain] node[pos=.5]{6}(f);
    \draw (a) edge[chain] node[pos=.5]{10}(b);
    \draw (b) edge[chain] node[pos=.5]{0}(c);
    \draw (c) edge[chain] node[pos=.5]{2}(d);
    \draw (b) edge[chain] node[pos=.58]{13}(f);
    \draw (c) edge[chain] node[pos=.38]{11}(f);
    \draw (d) edge[chain] node[pos=.42]{14}(f);
    \draw (e) edge[chain] node[pos=.5]{1}(f);
    \draw [dashed, color=gray](c)-- node[pos=.42, text=gray]{4}(e);
    \draw [dashed, color=gray](b)-- node[pos=.38, text=gray]{12}(e);
    \draw [dashed, color=gray](b)-- node[pos=.42, text=gray]{5}(d);
    \draw [dashed, color=gray](d)-- node[pos=.5, text=gray]{9}(e);
  \end{tikzpicture}
  $\Longleftrightarrow$
  \begin{tikzpicture}[scale=1.8]
    \path (0,0) coordinate (m);
    \tikzstyle{every node}=[circle,fill=white,inner sep=1.3pt]
    \path (m)+(90:1) node[emitter] (a){};
    \path (m)+(30:1) node[defnode] (b){};
    \path (m)+(330:1) node[defnode] (c){};
    \path (m)+(270:1) node[defnode] (d){};
    \path (m)+(210:1) node[defnode] (e){};
    \path (m)+(150:1) node[emitter] (f){};
    \tikzstyle{every node}=[circle,fill=white,inner sep=.6pt,font=\footnotesize]
    \tikzstyle{every path}=[shorten >= 5pt, shorten <= 5pt]
    \tikzstyle{chain}=[very thick, shorten <= 5pt, shorten >= 5pt, green]
    \draw (a)-- node[pos=.5]{10}(b);
    \draw (a)-- node[pos=.42]{7}(c);
    \draw (a) edge[chain,<-] node[pos=.38]{3}(d);
    \draw (a)-- node[pos=.58]{8}(e);
    \draw (a)-- node[pos=.5]{6}(f);
    \draw (b) edge[chain,->] node[pos=.5]{0}(c);
    \draw (b)-- node[pos=.42]{5}(d);
    \draw (b)-- node[pos=.38]{12}(e);
    \draw (b)-- node[pos=.58]{13}(f);
    \draw (c) edge[chain,->] node[pos=.5]{2}(d);
    \draw (c)-- node[pos=.42]{4}(e);
    \draw (c)-- node[pos=.38]{11}(f);
    \draw (d)-- node[pos=.5]{9}(e);
    \draw (d)-- node[pos=.42]{14}(f);
    \draw (e) edge[chain,->] node[pos=.5]{1}(f);

    \tikzstyle{every path}=[shorten <= 5pt, very thick, ->, green]
    \path (a)+(-30:.3) coordinate (ab);
    \path (a)+(-60:.3) coordinate (ac);
    \path (a)+(-120:.3) coordinate (ae);
    \path (a)+(-150:.3) coordinate (af);
    \path (f)+(30:.3) coordinate (fa);
    \path (f)+(0:.3) coordinate (fb);
    \path (f)+(-30:.3) coordinate (fc);
    \path (f)+(-60:.3) coordinate (fd);
    \draw (a)--(ab);
    \draw (a)--(ac);
    \draw (a)--(ae);
    \draw (a)--(af);
    \draw (f)--(fa);
    \draw (f)--(fb);
    \draw (f)--(fc);
    \draw (f)--(fd);
  \end{tikzpicture}
}
\readlist\arg{10 7 3 8 6 0 5 12 13 2 4 11 9 14 1}
\hexafull
\caption{\label{fig:forward-fireworks-spanner}Example of a forward fireworks spanner (left) and the journeys from which it is constructed (right).
}
\end{figure}

\begin{theorem}
\label{theorem:3quarter}
Forward fireworks spanners have at most $3\m /4 + O(n)$ edges.
\end{theorem} 

\begin{proof}
Let $S_T^-$ be a forward fireworks spanner based on a set of in-trees ${\cal T}^-$. Each non-emitter in ${\cal T}^-$ has only one out-arc which becomes one edge in $S_T^-$, thus overall ${\cal T}^-$ contributes less than $n$ edges to $S_T^-$. Now, every emitter has an edge to every other vertex that is included in $S_T^-$, and there are at most $n/2$ emitters in ${\cal T}^-$ by Lemma~\ref{lemma:half_n_emitters}. This contributes at most $(n/2)(n-1)$ edges. Note that the edges between emitters are selected twice but should be counted only once. Thus in the end, there are at most $(n/2)(n-1) - {(n/2) \choose 2} = 3n^2/8 -n/4 = 3\m /4 + n/8$ edges, plus the $O(n)$ edges contributed by the in-trees.
\end{proof}

Before moving to Section~\ref{sec:reverse}, we establish a small technical lemma about the structure of the in-trees, which will be used in Section~\ref{sec:iter_deleg}.

\begin{lemma}
\label{lemma:journeythroughminus}
Every non-emitter vertex $u$ in $\mathcal{T}^-$ can reach another vertex $v$ in the same in-tree using a journey of length {\em at most two} that arrives at $v$ through $e^-(v)$.
\end{lemma}

\begin{proof}
For each non-emitter vertex $u$, there exists a $w$ such that $(u,w) \in E_T^-$. Either $(u,w) \in E^-$ and then $\{u,w\}=e^-(w)$, or $(u,w)$ has been obtained by flipping $(w,u) \in E^-$. In the latter case, another vertex $v$ must exist such that $(w,v) \in E^- \cap E_T^-$ and, hence, $\{w,v\}=e^-(v)$.
\end{proof}

\subsection{Backward Fireworks}
\label{sec:reverse}

A symmetrical concept of fireworks can be defined based on the edges $\{u,v\} = e^+(u)$ of a temporal clique $\G=(G,\lambda)$. All arguments developed in the context of forward fireworks can be adapted in a symmetrical way, so we will omit most of the details.
First, we build a directed graph ${G^+=(V,E^+)}$ such that $(u,v)\in E^+$ iff $\{u,v\} = e^+(u)$, except that, if $e^+(u) = e^+(v)$ for some $u$ and $v$, only one of the arcs is included (chosen arbitrarily). $E^+$ induces a disjoint set of {\em in-trees}. 

We transform $G^+$ into a disjoint set ${\cal T}^+= (V,E^+_T)$ of {\em out-trees} 
as follows. Let $E^+_T$ be initialized as a copy of $E^+$. 
For every $v$ with indegree at least $2$ in $G^+$, let $(u_1,v), ..., (u_{\ell},v)$ be its in-arcs with $(u_{\ell},v)$ being the one with the {\em smallest} label.
For every $i < \ell$, if $u_i$ is a source vertex, then flip the direction of $(u_i,v)$ in $E^+_T$ (\ie replace $(u_i,v)$ by $(v,u_i)$ in $E^+_T$); otherwise remove $(u_i,v)$ from $E^+_T$. Let ${\cal T}^+= (V,E^+_T)$ be the resulting set of out-trees ${\cal T}_1^+, ..., {\cal T}_k^+$. 

Each out-tree in ${\cal T}^+$ contains only one source which we call a {\em collector}.
The collector $s$ of an out-tree can reach all of the other vertices in this tree by journeys that leave $s$ through its edge $e^+(s)$, thereby guaranteeing that every other vertex that reaches $s$ can subsequently reach all other vertices in the tree. Finally, we can build a temporal spanner $S_T^+=\{\{u,v\} \in E:(u,v) \in {\cal T}^+\} \cup \{\{u,v\} \in E: u$ is a collector$\}$ which we call a {\em backward fireworks spanner}. An example of a backward fireworks spanner is given in Figure~\ref{fig:backward-fireworks-spanner}. 
The proofs of the following four lemmas and theorems are symmetrical to the arguments in Section~\ref{subsec:forward_fw}.

\begin{lemma}
\label{lemma:half_n_collectors}
The number of collectors in ${\cal T}^+$ is at most $n/2$
\end{lemma}


\begin{theorem}
\label{th:backward-valid}
$S_T^+$ is a temporal spanner of the temporal clique $\G$.
\end{theorem}

\begin{theorem}
\label{theorem:3quarter-back}
Backward fireworks spanners have at most $3\m/4 + O(n)$ edges.
\end{theorem} 

\begin{figure}[htb]
\centering
\newcommand{\hexafull}{
  \begin{tikzpicture}[scale=1.9]
    \path (0,0) coordinate (m);
    \path (m)+(90:1) node[defnode] (a){};
    \path (m)+(30:1) node[defnode] (b){};
    \path (m)+(330:1) node[defnode] (c){};
    \path (m)+(270:1) node[defnode] (d){};
    \path (m)+(210:1) node[defnode] (e){};
    \path (m)+(150:1) node[defnode] (f){};
    \tikzstyle{every node}=[circle,fill=white,inner sep=1.3pt]
    \tikzstyle{every path}=[shorten >= 5pt, shorten <= 5pt]
    \tikzstyle{chain}=[very thick, shorten <= 5pt, shorten >= 5pt]
    \draw (a) edge[chain] node[pos=.42]{7}(c);
    \draw (a) edge[chain] node[pos=.38]{3}(d);
    \draw (a) edge[chain] node[pos=.58]{8}(e);
    \draw (a) edge[chain] node[pos=.5]{6}(f);
    \draw (a) edge[chain] node[pos=.5]{10}(b);
    \draw (b) edge[chain] node[pos=.5]{0}(c);
    \draw (c) edge[chain] node[pos=.5]{2}(d);
    \draw (c) edge[chain] node[pos=.42]{4}(e);
    \draw (c) edge[chain] node[pos=.38]{11}(f);
    \draw (b) edge[chain] node[pos=.38]{12}(e);
    \draw (c) edge[chain] node[pos=.38]{11}(f);
    \draw (d) edge[chain] node[pos=.42]{14}(f);
    \draw [dashed, color=gray](b)-- node[pos=.58, text=gray]{13}(f);
    \draw [dashed, color=gray](b)-- node[pos=.42, text=gray]{5}(d);
    \draw [dashed, color=gray](d)-- node[pos=.5, text=gray]{9}(e);
    \draw [dashed, color=gray](e)-- node[pos=.5, text=gray]{1}(f);
  \end{tikzpicture}
  $\Longleftrightarrow$
  \begin{tikzpicture}[scale=1.9]
    \path (0,0) coordinate (m);
    \path (m)+(90:1) node[receiver] (a){};
    \path (m)+(30:1) node[defnode] (b){};
    \path (m)+(330:1) node[receiver] (c){};
    \path (m)+(270:1) node[defnode] (d){};
    \path (m)+(210:1) node[defnode] (e){};
    \path (m)+(150:1) node[defnode] (f){};
    \tikzstyle{every node}=[circle,fill=white,inner sep=1.3pt]
    \tikzstyle{every path}=[shorten >= 5pt, shorten <= 5pt]
    \draw (a)-- node[pos=.5]{10}(b);
    \draw (a)-- node[pos=.42]{7}(c);
    \draw (a)-- node[pos=.38]{3}(d);
    \draw (a)-- node[pos=.58]{8}(e);
    \draw (a)-- node[pos=.5]{6}(f);
    \draw (b)-- node[pos=.5]{0}(c);
    \draw (b)-- node[pos=.42]{5}(d);
    \draw (b)-- node[pos=.38]{12}(e);
    \draw (b)-- node[pos=.58]{13}(f);
    \draw (c)-- node[pos=.5]{2}(d);
    \draw (c)-- node[pos=.42]{4}(e);
    \draw (c)-- node[pos=.38]{11}(f);
    \draw (d)-- node[pos=.5]{9}(e);
    \draw (d)-- node[pos=.42]{14}(f);
    \draw (e)-- node[pos=.5]{1}(f);
    
    \tikzstyle{chain}=[ultra thick, shorten <= 5pt, shorten >= 5pt, red]
    \draw (a) edge[chain,->] node[pos=.5]{10}(b);
    \draw (b) edge[chain,->] node[pos=.38]{12}(e);
    \draw (c) edge[chain,->] node[pos=.38]{11}(f);
    \draw (d) edge[chain,<-] node[pos=.42]{14}(f);

    \tikzstyle{every path}=[shorten <= 5pt, ultra thick, <-, red]
    \path (a)+(-150:.3) coordinate (af);
    \path (a)+(-120:.3) coordinate (ae);
    \path (a)+(-90:.3) coordinate (ad);
    \path (a)+(-60:.3) coordinate (ac);
    \path (c)+(120:.3) coordinate (ca);
    \path (c)+(90:.3) coordinate (cb);
    \path (c)+(210:.3) coordinate (cd);
    \path (c)+(180:.3) coordinate (ce);
    \draw (a)--(af);
    \draw (a)--(ae);
    \draw (a)--(ad);
    \draw (a)--(ac);
    \draw (c)--(ca);
    \draw (c)--(cb);
    \draw (c)--(cd);
    \draw (c)--(ce);
  \end{tikzpicture}
}
\readlist\arg{10 7 3 8 6 0 5 12 13 2 4 11 9 14 1}
\hexafull
\caption{\label{fig:backward-fireworks-spanner}Example of a backward fireworks spanner (left) and the journeys from which it is constructed (right).}
\end{figure}

\begin{lemma}
\label{lemma:journeythroughmax}
Every non-collector vertex $v$ in ${\cal T}^+$ can reach another vertex $v'$ in the same out-tree using a journey of length {\em at most two} that leaves $v'$ through $e^+(v')$.
\end{lemma}

\subsection{Bidirectional Fireworks}
\label{sec:bidirectional}

A forward fireworks spanner makes it possible to identify a subset of vertices, the {\em emitters}, such that every vertex can reach at least one emitter $u$ through $e^-(u)$ and $u$ can reach every other vertex afterwards {\em through a single edge}. Similarly, a backward fireworks spanner makes it possible to identify a subset of vertices, the {\em collectors}, such that every vertex can be reached by at least one collector $v$ through $e^+(v)$ and $v$ can be reached by every other vertex before this {\em through a single edge}. Combining both ideas, we can define a sparser spanner in which we do not need to include {\em all} of the edges incident to emitters and collectors, but only the edges {\em between} emitters and collectors (plus, of course, the edges used for reaching an emitter and for being reached by a collector).

Precisely, let ${\cal T}^-$ be the disjoint set of in-trees obtained during the construction of a forward fireworks spanner (see Figure~\ref{fig:minchain}), 
and let ${\cal T}^+$ be the disjoint set of out-trees obtained during the construction of a backward fireworks spanner. 
Let $X^-$ be the set of emitters (one per in-tree in ${\cal T}^-$) and let $X^+$ be the set of collectors (one per out-tree in ${\cal T}^+$). The two sets can overlap, as a vertex may happen to be both an emitter in some tree in ${\cal T}^-$ and a collector in some tree in ${\cal T}^+$. 
Let $H=(X^- \cup X^+, E_H)$ be the graph such that $E_H=\{\{u,v\} \in E : u \in X^-, v \in X^+\}$; in other words, $H$ is the subgraph of $G$ that connects all emitters with all collectors. Finally, let $S=\{\{u,v\}:(u,v) \in {\cal T}^- \cup {\cal T}^+\} \cup E_H$. We call $S$ a bidirectional fireworks spanner (or simply a {\em fireworks spanner}). An illustration is given in Figure~\ref{fig:connected-fireworks-spanner}.

\begin{figure}[htb]
\centering
\newcommand{\hexafull}{
  \begin{tikzpicture}[scale=1.9]
    \path (0,0) coordinate (m);
    \path (m)+(90:1) node[defnode] (a){};
    \path (m)+(30:1) node[defnode] (b){};
    \path (m)+(330:1) node[defnode] (c){};
    \path (m)+(270:1) node[defnode] (d){};
    \path (m)+(210:1) node[defnode] (e){};
    \path (m)+(150:1) node[defnode] (f){};
    \tikzstyle{every node}=[circle,fill=white,inner sep=.4pt,font=\footnotesize]
    \tikzstyle{every path}=[shorten >= 5pt, shorten <= 5pt]
    \tikzstyle{chain}=[very thick, shorten <= 5pt, shorten >= 5pt]
    \draw (a) edge[chain] node[pos=.42]{7}(c);
    \draw (a) edge[chain] node[pos=.38]{3}(d);
    \draw (a) edge[chain] node[pos=.5]{6}(f);
    \draw (b) edge[chain] node[pos=.5]{0}(c);
    \draw (c) edge[chain] node[pos=.5]{2}(d);
    \draw (e) edge[chain] node[pos=.5]{1}(f);
    \draw (a) edge[chain] node[pos=.5]{10}(b);
    \draw (b) edge[chain] node[pos=.38]{12}(e);
    \draw (c) edge[chain] node[pos=.38]{11}(f);
    \draw (d) edge[chain] node[pos=.42]{14}(f);
    \draw [dashed, color=gray](a)-- node[pos=.58, text=gray]{8}(e);
    \draw [dashed, color=gray](b)-- node[pos=.58, text=gray]{13}(f);
    \draw [dashed, color=gray](b)-- node[pos=.42, text=gray]{5}(d);
    \draw [dashed, color=gray](c)-- node[pos=.42, text=gray]{4}(e);
    \draw [dashed, color=gray](d)-- node[pos=.5, text=gray]{9}(e);
  \end{tikzpicture}
  $\Longleftrightarrow$
  \begin{tikzpicture}[scale=1.9]
    \path (0,0) coordinate (m);
    \path (m)+(90:1) node[emitter_receiver] (a){};
    \begin{scope}[on background layer, rotate=30]
    \fill[green] (a.base) ([xshift = 0mm]a.east) arc (0:180:0.7mm)
      -- cycle;
    \fill[red] (a.base) ([xshift = 0pt]a.west) arc (180:360:0.72mm)
      -- cycle;
  \end{scope}
    \path (m)+(30:1) node[defnode] (b){};
    \path (m)+(330:1) node[receiver] (c){};
    \path (m)+(270:1) node[defnode] (d){};
    \path (m)+(210:1) node[defnode] (e){};
    \path (m)+(150:1) node[emitter] (f){};
    \tikzstyle{every node}=[circle,fill=white,inner sep=.4pt,font=\footnotesize]
    \tikzstyle{every path}=[shorten >= 5pt, shorten <= 5pt]
    \tikzstyle{chain}=[ultra thick, shorten <= 5pt, shorten >= 5pt, green]
    \draw (a)-- node[pos=.5]{10}(b);
    \draw (a)-- node[pos=.42]{7}(c);
    \draw (a) edge[chain,<-] node[pos=.38]{3}(d);
    \draw (a)-- node[pos=.58]{8}(e);
    \draw (a)-- node[pos=.5]{6}(f);
    \draw (b) edge[chain,->] node[pos=.5]{0}(c);
    \draw (b)-- node[pos=.42]{5}(d);
    \draw (b)-- node[pos=.38]{12}(e);
    \draw (b)-- node[pos=.58]{13}(f);
    \draw (c) edge[chain,->] node[pos=.5]{2}(d);
    \draw (c)-- node[pos=.42]{4}(e);
    \draw (c)-- node[pos=.38]{11}(f);
    \draw (d)-- node[pos=.5]{9}(e);
    \draw (d)-- node[pos=.42]{14}(f);
    \draw (e) edge[chain,->] node[pos=.5]{1}(f);

    \tikzstyle{every path}=[shorten <= 5pt, ultra thick, ->, green]
    \path (a)+(-60:.3) coordinate (ac);
    \path (f)+(30:.3) coordinate (fa);
    \draw (a)--(ac);
    \draw (f)--(fa);
    
    \tikzstyle{chain}=[ultra thick, shorten <= 5pt, shorten >= 5pt, red]
    \draw (a) edge[chain,->] node[pos=.5]{10}(b);
    \draw (b) edge[chain,->] node[pos=.38]{12}(e);
    \draw (c) edge[chain,->] node[pos=.38]{11}(f);
    \draw (d) edge[chain,<-] node[pos=.42]{14}(f);

    \tikzstyle{every path}=[shorten <= 5pt, ultra thick, <-, red]
    \path (a)+(-150:.3) coordinate (af);
    \path (c)+(120:.3) coordinate (fd);
    \draw (a)--(af);
    \draw (c)--(fd);
  \end{tikzpicture}
}
\readlist\arg{10 7 3 8 6 0 5 12 13 2 4 11 9 14 1}
\hexafull
\caption{\label{fig:connected-fireworks-spanner}A bidirectional fireworks spanner (left) and the journeys from which it is constructed (right). The forward component and the emitters are depicted in (light) green; the backward component and the collectors are depicted in (darker) red. The top vertex is both an emitter and a collector.
}
\end{figure}

\begin{theorem}
  \label{th:bidirectional}
  $S$ is a temporal spanner of $\G$.
\end{theorem}
\begin{proof}
  Every non-emitter vertex can reach at least one emitter $u$ through $e^-(u)$. Every emitter can reach {\em all} collectors afterwards. Every vertex can be reached by a collector $v$ through $e^+(v)$.
\end{proof}

\begin{theorem}
Bidirectional fireworks spanners have at most $\m/2 + O(n)$ edges.
\end{theorem}

\begin{proof}
The number of edges in ${\cal T}^-$ and ${\cal T}^+$ is linear in $n$.
By Lemma~\ref{lemma:half_n_emitters}, the number of emitters is at most $n/2$, and so is the number of collectors by Lemma~\ref{lemma:half_n_collectors}. Some vertices may be both emitter and collector; however, the number of edges is maximized when $X^-$ and $X^+$ are disjoint, \ie $H$ is a complete bipartite graph with $n/2$ vertices in each part. Thus, the spanner contains at most $n^2/4 = \m/2 +n/4$ edges plus the $O(n)$ edges of ${\cal T}^-$ and ${\cal T}^+$.
\end{proof}

\section{Recursing or Sparsifying}
\label{sec:iter_deleg}

After applying the fireworks technique, one is left with a residual instance (a spanner of the input clique $\G$ on vertices $V$) made of all the edges between emitters $X^-$ and collectors $X^+$, together with all the edges corresponding to the arcs of ${\cal T}^-$ and ${\cal T}^+$, these edges being denoted $S^-$ and $S^+$. 
Depending on the properties of this residual instance, the algorithm may dismount it (and recurse on a subset of vertices), or it may continue to sparsify it using different techniques. The global algorithm is summarized in \cref{fig:global_algorithm}.
To simplify the notation in the following, we use variables $\G$ and $V$ to refer to the residual instance in the current step of the recursion.

The following two cases are considered: Either $X^- \cup X^+ \ne V$ (Case~1) or $X^- \cup X^+ = V$ (Case~2).

\setcounter{case}{0}
\begin{case}[$X^- \cup X^+ \ne V$]
\label{case:em_rec_neq_V}
In this configuration, at least one vertex $v$ is neither emitter nor collector.
By Lemma~\ref{lemma:journeythroughminus}, there exists a journey of length at most two from $v$ that arrives at some vertex $u \ne v$ through $e^-(u)$. Similarly, by Lemma~\ref{lemma:journeythroughmax}, there is a journey of length at most two from some vertex $w \ne v$ to $v$, leaving $w$ through $e^+(w)$. As a result, $v$ is $2$-hop dismountable (see Section~\ref{sec:delegation}). One can thus select the corresponding edges (at most four) for future inclusion in the spanner and recurse on $\G[V\setminus v]$; that is, re-apply the fireworks technique from scratch to this smaller graph. Repeating this strategy, either the recursion keeps entering Case~1 and dismounts the graph entirely, or it eventually enters Case~2.
\end{case}

\begin{case}[$X^- \cup X^+ = V$]
\label{case:em_rec_eq_V}
Both $X^-$ and $X^+$ have size at most $n/2$ (Lemma~\ref{lemma:half_n_emitters} and~\ref{lemma:half_n_collectors}), thus if their union is $V$, then both sets must be disjoint and of size exactly $n/2$.
As a result, the graph which connects vertices of $X^-$ to vertices of $X^+$ (called $H$ in Section~\ref{sec:fireworks}) is a complete bipartite graph. In fact, $H$ possesses even more structure; in particular, both $S^-$ and $S^+$ are perfect matchings -- by contradiction, if this is not the case, then at least one of the in-trees (out-trees) contains more than one edge, resulting in strictly less than $n/2$ emitters (collectors). 
Furthermore, every vertex is either an emitter or a collector, thus each of these edges connects an emitter with a collector, implying that the current residual instance $\G$ is $H$ itself. Now, recall that every edge in $S^-$ is locally minimum for the corresponding emitter (Lemma~\ref{fact:Tminus}.5), and every edge in $S^+$ is locally maximum for the corresponding collector. We thus have the following stronger property.

\begin{lemma}
  \label{lem:even-more}
If the minimum edge of an emitter is not also the minimum edge of the corresponding collector in $H$, then the residual instance is $2$-hop dismountable. The same holds if the maximum edge of a collector is not also the maximum edge of the corresponding emitter in $H$.
\end{lemma}

\begin{proof}
  Let us prove this for minimum edges (a symmetrical argument applies for maximum edges). Consider an emitter $u$ whose minimum edge $\{u,v\}=e^-(u)$ leads to collector $v$ such that $\{u,v\}\neq e^-(v)$ in the bipartite graph. Then an edge with a smaller label exists between $v$ and another emitter $u'$ such that $\{u',v\}= e^-(v)$, which creates a $2$-hop journey from $u'$ to $u$ arriving through $e^-(u)$, implying that $u'$ can delegate its emissions to $u$. 
Moreover, emitters and collectors are disjoint, thus $u'$ is not a collector. As a result, $u'$ already delegates its receptions to a collector (through a single edge), thus $u'$ is $2$-hop dismountable.
\end{proof}

\noindent There are two subcases of Case~2 depending on whether or not the hypothesis of Lemma~\ref{lem:even-more} holds.

\begin{subcase}[$H$ contains a 2-hop dismountable vertex $v$]
\label{subcase:dismountable}
The algorithm 2-hop dismounts $v$, the corresponding edges (at most four) are selected for future inclusion in the spanner, and the algorithm recurses on $\G[V\setminus v]$.
\end{subcase}
\begin{subcase}[The edges of the matchings $S^-$ and $S^+$ are minimum (resp. maximum) {\em on both sides}]
\label{subcase:not-dismountable}
An example of this subcase is given in Figure~\ref{fig:bipartite_matchings}.
\end{subcase}

In summary, either the algorithm recurses until the input clique is fully dismounted (through Case~1 and Subcase~\ref{subcase:dismountable}), resulting in a spanner of size $O(n)$ (\cref{fact:k-hop-number}), or the {\em recursion stops} and the residual instance is sparsified further by a dedicated procedure for Subcase~\ref{subcase:not-dismountable}, described now.


\end{case}

\begin{figure}[htb]
\centering
\newcommand{\hexafull}{
\begin{tikzpicture}[scale=2.4]
    \path (0,0) coordinate (m);
    \path (m)+(90:1) node[defnode] (a){};
    \path (m)+(45:1) node[defnode] (b){};
    \path (m)+(0:1) node[defnode] (c){};
    \path (m)+(-45:1) node[defnode] (d){};
    \path (m)+(-90:1) node[defnode] (e){};
    \path (m)+(-135:1) node[defnode] (f){};
    \path (m)+(180:1) node[defnode] (g){};
    \path (m)+(135:1) node[defnode] (h){};
    \tikzstyle{every node}=[circle,fill=white,inner sep=1.3pt,font=\footnotesize]
    \tikzstyle{every path}=[shorten >= 5pt, shorten <= 5pt]
    \tikzstyle{chain}=[ultra thick, shorten <= 5pt, shorten >= 5pt, green]
	\draw (a)-- node[pos=.5]{26}(b);
    \draw (a) -- node[pos=.2]{2}(c);
    \draw (a) -- node[pos=.5]{19}(d);
    \draw (a) -- node[pos=.2]{13}(e);
    \draw (a) -- node[pos=.5]{27}(f);
    \draw (a) -- node[pos=.82]{17}(g);
    \draw (a) -- node[pos=.5]{3}(h);
    \draw (b) -- node[pos=.5]{0}(c);
    \draw (b) -- node[pos=.2]{5}(d);
    \draw (b) -- node[pos=.5]{15}(e);
    \draw (b) -- node[pos=.2]{12}(f);
    \draw (b) -- node[pos=.5]{11}(g);
    \draw (b) -- node[pos=.8]{1}(h);
    \draw (c) -- node[pos=.5]{23}(d);
    \draw (c) -- node[pos=.17]{10}(e);
    \draw (c) -- node[pos=.5]{9}(f);
    \draw (c) -- node[pos=.8]{8}(g);
    \draw (c) -- node[pos=.5]{24}(h);
    \draw (d) -- node[pos=.42]{20}(e);
    \draw (d) -- node[pos=.2]{4}(f);
    \draw (d) -- node[pos=.5]{25}(g);
    \draw (d) -- node[pos=.2]{22}(h);
    \draw (e) -- node[pos=.5]{21}(f);
    \draw (e) -- node[pos=.2]{6}(g);
    \draw (e) -- node[pos=.5]{14}(h);
    \draw (f) -- node[pos=.5]{18}(g);
    \draw (f) -- node[pos=.2]{7}(h);
    \draw (g) -- node[pos=.5]{16}(h);
  \end{tikzpicture}
  $\to\quad$
  \begin{tikzpicture}[scale=1.5]
  \path (0,0) coordinate (m);
  \path (-1, -2) node[emitter] (G1){};
  \path (-1, -1) node[emitter] (G2){};
  \path (-1, 0) node[emitter] (G3){};
  \path (-1, 1) node[emitter] (G4){};
  \path (1, -2) node[receiver] (R1){};
  \path (1, -1) node[receiver] (R2){};
  \path (1, 0) node[receiver] (R3){};
  \path (1, 1) node[receiver] (R4){};
    \tikzstyle{every node}=[circle,fill=white,inner sep=0.4pt,font=\footnotesize]
    \tikzstyle{every path}=[shorten >= 5pt, shorten <= 5pt]
    \tikzstyle{chain}=[ultra thick, shorten <= 5pt, shorten >= 5pt, green]
    
    \draw (G1) edge[chain, color=red] node[pos=.2]{21}(R1);
    \draw (G1) edge[chain, color=green] node[pos=.3]{4}(R2);
    \draw (G1) -- node[pos=.4]{9}(R3);  
    \draw (G1) -- node[pos=.18]{12}(R4);  
    
    \draw (G2) edge[chain, color=green] node[pos=.4]{6}(R1);  
    \draw (G2) edge[chain, color=red] node[pos=.2]{25}(R2);  
    \draw (G2) -- node[pos=.25]{8}(R3);  
    \draw (G2) -- node[pos=.15]{11}(R4);  
    
    \draw (G3) -- node[pos=.15]{14}(R1);  
    \draw (G3) -- node[pos=.25]{22}(R2);
    \draw (G3) edge[chain, color=red] node[pos=.2]{24}(R3);
    \draw (G3) edge[chain, color=green] node[pos=.4]{1}(R4);
    
    \draw (G4) -- node[pos=.18]{13}(R1);
    \draw (G4) -- node[pos=.4]{19}(R2);
    \draw (G4) edge[chain, color=green] node[pos=.3]{2}(R3);
    \draw (G4) edge[chain, color=red] node[pos=.2]{26}(R4);
  \end{tikzpicture}
}
\readlist\arg{10 7 3 8 6 0 5 12 13 2 4 11 9 14 1}
\hexafull
\caption{\label{fig:bipartite_matchings}Temporal clique for which the output of the fireworks technique results in Subcase~\ref{subcase:not-dismountable} in which every vertex is either an emitter or a collector and the edges of the matchings are minimum ({\it resp.} maximum) on both sides. The minimum edges are depicted in light green and maximum edges in dark red.}
\end{figure}

\subsection{Sparsifying the Residual Instance}

For simplicity, the sparsification of the residual instance is considered as a separate problem. The input is a labelled complete bipartite graph $B=(X^-,X^+,E_B)$ where $X^-$ is the set of emitters, $X^+$ is the set of collectors, and the labels are inherited from $\G$. There are two perfect matchings $S^-$ and $S^+$ in $B$ such that the labels of the edges in $S^-$ (resp. $S^+$) are minimum (resp. maximum) locally to both of their endpoints (Lemma~\ref{lem:even-more}). 
The objective is to remove as many edges as possible from $E_B$, while preserving $S^-$, $S^+$, and the fact that every emitter can reach {\em all} collectors by a journey. Indeed, these three properties ensure temporal connectivity of the graph (using the same arguments as in Theorem~\ref{th:bidirectional}).

While both $S^-$ and $S^+$ are matchings, our algorithm only exploits this property with respect to $S^+$ as follows. 

\begin{lemma}
  \label{fact:E2E}
  In $B$, if an emitter $u$ can reach another emitter $v$, then it can reach the collector $w$ corresponding to $v$ by adding the edge $(v,w) \in S^+$ to the journey from $u$ to $v$.
\end{lemma}

\begin{proof}
Edge $(v, w) \in S^+$, thus its label is locally maximum to $v$ in $B$ (\cref{lem:even-more}). This implies that edge $(v, w)$ may be added to any journey to $v$, resulting in a journey to $w$. Observe that, in the particular case that $(u,w)$ is in $S^-$, then this edge is already in the journey from $u$ to $v$ (with the consequence that an edge is saved).
\end{proof}

This property makes it possible to reduce the task of reaching some collectors to that of reaching their corresponding emitters in $S^+$. 
It is however impossible for an emitter $u$ to make a complete delegation to another emitter $v$, because the existence of a journey from $u$ to $v$ arriving through $e^-(v)$ would contradict the fact that $S^-$ is also a matching. For this reason, when a journey from emitter $u$ arrives at emitter $v$, some of $v$'s edges have already disappeared. Nevertheless, the algorithm exploits such {\em partial} delegations, while paying extra edges for the missed opportunities (contained within a logarithmic factor). This is done by means of an {\em iterative} procedure called {\em layered delegations}, described in the remainder of this section. Note the term iterative, not recursive; from now on, the instance has a fixed vertex set and it is sparsified until the final bound is reached.

\subsubsection{Layered Delegations}

The algorithm proceeds by {\em eliminating} half of the emitters in each step $j$, while selecting a set $S_j$ of edges for inclusion in the spanner, such that the eliminated emitters can reach all collectors using either single edges or indirect journeys through other emitters (partial delegations). The set of non-eliminated emitters at step $j$ (called {\em alive}) is denoted by $X^-_j$, with $X^-_1=X^-$. The set of collectors $X^+$ is invariant over the execution. We denote by $k=n/2$ the initial degree of the emitters in $B$ (one edge shared with each collector), and by $e^i(v)$ the edge with the $i^{th}$ smallest label  (label of {\em rank} $i$) locally to a vertex $v$, in particular $e^1(v)=e^-(v)$ and $e^k(v)=e^+(v)$.

The $k$ ranks are partitioned into subintervals of doubling size ${\cal I}_j=[2^{j+2}-7, 2^{j+3}-8]$, where $j$ denotes the current step of the iteration, ranging from $1$ to $\log_2 k-3$. For simplicity, assume that $k$ is a power of two; we explain below how to adapt the algorithm when this is not the case (see Remark~\ref{rem:power-of-2}). For example, if $k=128$, then ${\cal I}_1=[1,8], {\cal I}_2=[9,24], {\cal I}_3=[25,56],$ and ${\cal I}_4=[57,120]$.
Computation step $j$ is made with respect to the subgraph $B_j=(X^-_j,X^+,E_j)$ where $E_j=\{e^i(v) \in E_B: i\in {\cal I}_j, v\in X^-_j\}$, namely the edges of the currently alive emitters, whose ranks are in the interval ${\cal I}_j$.

\begin{lemma}
  \label{lem:half-eliminated}
  In each step $j$, 
  $X^-_j$ can be split into two sets $X_a$ and $X_b$ such that $|X_a| \ge |X_b|$ and every vertex in $X_a$ can reach a vertex in $X_b$ through a $2$-hop journey (within $B_j$).
\end{lemma}

\begin{proof}
This proof is illustrated in Figure~\ref{fig:proof_lemma_7} for the particular case that $j=1$ (first step).   We will show that the average degrees of collectors in $B_j$ forces the existence of sufficiently many two-hop journeys among emitters.
  To start, observe that if a collector $v$ shares an edge with $d$ emitters in $B_j$ (we say that these emitters {\em meet} at $v$), then the emitter whose edge with $v$ has the largest label can be reached by the $d-1$ other emitters through two-hop journeys. Thus, the proof proceeds by showing that, in each step $j$, the distribution of degrees over collectors forces the existence of sufficiently many such meetings among emitters. 
Here, the size of the first interval ${\cal I}_j$ matters, because if one starts with intervals of size only $2$ or $4$ (say), then the density of edges remains insufficient for the argument to apply, and starting with $8$ (or any larger constant power of two) suffices, while not impacting the asymptotic cost, as we will see. Also observe that the doubling size of the rank intervals cancels out the halving size of $X_j^-$ over the steps, leading to an average degree for collectors (in the considered range of ranks) that remains constant over the steps (namely $8$).

The calculation relative to step $j$ is itself based on an iterative argument that one should be careful not to confuse with the outer loop where $j$ varies. Thus, keeping $j$ fixed for the rest of the proof, $X_a$ and $X_b$ are built iteratively as follows: find the collector $c$ with highest degree and add all the corresponding emitters to $X_a$ except the one whose edge with $c$ has largest label, which is added to $X_b$; eliminate these emitters and repeat until $X_a \ge X_j^-/2$; then add the remaining emitters to $X_b$. To see why this works, observe that the average degree of $8$ for collectors implies that at least one collector has degree at least $8$. In fact, this is true as long as the number of unprocessed emitters in $X_j^-$ (i.e. not yet in $X_a$ or $X_b$) is larger than $7/8 \cdot |X_j^-|$ (by the pigeon-hole principle). When $1/8$ of the emitters have been processed, we thus have at least $7/8$ of this $1/8$ fraction in $X_a$, i.e. $X_a$ contains at least $7/8\cdot1/8\cdot |X_j^-|$ (note that the emitters are not removed from $X_j^-$, just marked as processed and no longer counted for the degrees of collectors). An analogous argument implies that at least one collector has degree at least $7$ as long as the fraction of unprocessed emitters in $X_j^-$ remains above $6/8$, resulting in at least $6/7\cdot 1/8\cdot |X_j^-|$ more emitters in $X_a$ when the next threshold is attained, and resulting (again, for the sake of illustration) in at least $5/6 \cdot 1/8\cdot |X_j^-|$ more emitters when $5/8$ unprocessed emitters remain. By iterating this argument, the size of $X_a$ ends up being a fraction of $X_j^-$ at least equal to $1/8 \cdot (7/8 + 6/7 + 5/6 + 4/5 + 3/4 + 2/3 + 1/2) \simeq 0.66 \ge 0.5$.
\end{proof}

\begin{figure}[htb]
\centering
\newcommand{\hexafull}{
\begin{subfigure}[b]{0.25\textwidth}
  ~~
  \resizebox{\linewidth}{!}{
  \begin{tikzpicture}[xscale=.55,yscale=.43]
   \tikzstyle{every node}=[minimum size = .5cm]
  \path (0,0) coordinate (m);
  \path (-5, -5) node[emitter] (Gnn){};
  \path (-5, -4) node[emitter] (Gn){};
  \path (-5, -3) node[emitter] (G0){};
  \path (1, -5) node[receiver] (Rnn){};
  \path (1, -4) node[receiver] (Rn){};
  \path (1, -3) node[receiver] (R0){};
  \path (1, 1) node[receiver] (Rc){};
  \path (1, 2) node[invisible, rotate=90] (Rb){...};
  \path (1, 3) node[receiver] (Ra){};
  \path (-5, -2) node[emitter] (G1){};
  \path (-5, -1) node[invisible, rotate=90] (G2){...};
  \path (-5, 0) node[emitter] (G3){};
  \path (-5, 1) node[emitter] (G4){};
  \path (-5, 2) node[invisible, rotate=90] (G5){...};
  \path (-5, 3) node[emitter] (G6){};
  \path (-5, 4) node[emitter] (G7){};
  \path (-5, 5) node[emitter] (G8){};
  \path (-5, 6) node[invisible, rotate=90] (G9){...};
  \path (-5, 7) node[emitter] (G13){};
  \path (-5, 8) node[emitter] (G14){};
  \path (-5, 9) node[emitter] (G15){};
  \path (-5, 10) node[emitter] (G16){};
  \path (1, 10) node[receiver] (R1){};
  \path (1, 9) node[receiver] (R2){};
  \path (1, 8) node[receiver] (R5){};
  \path (1, 7) node[invisible, rotate=90] (R55){...};
  \path (1, 5) node[receiver] (R3){};
  \path (1, 4) node[receiver] (R4){};
  \path (1, -2) node[receiver] (R7){};
  \path (1, -1) node[invisible, rotate=90] (R8){...};
  \path (1, 6) node[receiver] (R10){};
  \path (1, 0) node[receiver] (R9){};
  
  \tikzstyle{every node}=[circle,fill=white,inner sep=1.3pt]
  \tikzstyle{every path}=[shorten >= 5pt, shorten <= 5pt]
  \tikzstyle{aux}=[dashed, gray, shorten <= 5pt, shorten >= 5pt]
  \tikzstyle{chain}=[thick, shorten <= 5pt, shorten >= 5pt]  
  
  \draw (G16) edge[chain, reverse directed] (R1);
  \draw (G15) edge[chain, directed] (R1);
  \draw (G14) edge[chain, directed] (R1);
  \draw (G13) edge[chain, directed](R1);
  \path (R1)+(-145:3) node[invisible] (v10){};
  \path (R1)+(-136:3) node[invisible] (v11){};
  \path (R1)+(-120:2) node[invisible, rotate=90] (v12){...};
  \path (R1)+(-118:3) node[invisible] (v13){};
  \draw [chain, reverse directedd] (R1) -- (v10);
  \draw [chain, reverse directedd] (R1) -- (v11);

  \draw (G8) edge[chain, reverse directed] (R3);
  \draw (G7) edge[chain, directed](R3);
  \draw (G6) edge[chain, directed] (R3);
  \path (R3)+(-152:3) node[invisible] (v70){};
  \path (R3)+(-146:3) node[invisible] (v71){};
  \path (R3)+(-137:2) node[invisible, rotate=90] (v72){...};
  \path (R3)+(-128:3) node[invisible] (v73){};
  \draw [chain, reverse directedd] (R3) -- (v70);
     
  \end{tikzpicture}
  }  
          \caption{Emitters ``meeting'' at collectors.
}
\label{subfig:adding_to_Vf}
\end{subfigure}
\qquad\qquad
\begin{subfigure}[b]{0.49\textwidth}
  \resizebox{\linewidth}{!}{
    \begin{tikzpicture}[xscale=.55,yscale=.43]
      \path (0,0) coordinate (m);
      \path (-5, -2) node[emitter] (G1){};
      \path (-5, -1) node[invisible, rotate=90] (G2){...};
      \path (-5, 0) node[emitter] (G3){};
      \path (-5, 1) node[emitter] (G4){};
      \path (-5, 2) node[emitter] (G5){};
      \path (-5, 3) node[invisible, rotate=90] (G6){...};
      \path (-5, 4) node[emitter] (G7){};
      \path (-5, 5) node[emitter] (G8){};
      \path (-5, 6) node[emitter] (G9){};
      \path (-5, 7) node[invisible, rotate=90] (G10){...};
      \path (-5, 8) node[emitter] (G11){};
      \path (-5, 9) node[invisible, rotate=90] (G12){...};
      \path (-5, 10) node[emitter] (G13){};
      \path (-5, 11) node[emitter] (G14){};
      \path (-5, 12) node[invisible, rotate=90] (G15){...};
      \path (-5, 13) node[emitter] (G16){};
      
      \draw[dashed] (-3,10.75) -- (-2.5,10.75) -- (-2.5,13.5) -- (-3, 13.5);
      \node[text width=2cm,align=center,draw, invisible,rectangle] at (0, 12)(){$\geq \frac{1}{8} \times \frac{7}{8} |X^-_j|$};
      \draw[dashed] (-3,7.75) -- (-2.5,7.75) -- (-2.5,10.25) -- (-3, 10.25);
      \node[text width=2cm,align=center,draw, invisible] at (0, 9)(){$\geq \frac{1}{8} \times \frac{6}{7} |X^-_j|$};
      \node[text width=1cm,align=center,draw, invisible, rotate=90] at (-2, 7)(){...};
      \draw[dashed] (-3,1.75) -- (-2.5,1.75) -- (-2.5,4.25) -- (-3, 4.25);
      \node[text width=2cm,align=center,draw, invisible] at (0, 3)(){$\geq \frac{1}{8} \times \frac{1}{2} |X^-_j|$};
      \node[text width=2cm,align=center,draw, invisible] at (6, 6)(){};
      
      \draw[dashed] (-3,-1.75) -- (-2.5,-1.75) -- (-2.5,.8) -- (-3, .8);
      \node[text width=.5cm,align=center,draw, invisible] at (-1, -.3)(){$\emptyset$};

      \draw [decorate,decoration={brace,amplitude=10pt,mirror,raise=4pt},yshift=0pt]
      (2,1.5) -- (2,13.5) node [text width =2.5cm, black,midway,xshift=2cm] {
        $|X_a| \geq |X^-_j|/2$}; 
      
    \end{tikzpicture}
  }  
  \caption{Calculation of the size of $X_a$ with respect to the set of emitters $X_j^-$.}
  \label{subfig:computing_deleted_nodes}
\end{subfigure}
}
\readlist\arg{10 7 3 8 6 0 5 12 13 2 4 11 9 14 1}
\hexafull
\caption{\label{fig:proof_lemma_7} Illustration of the method used in the proof of Lemma~\ref{lem:half-eliminated}.
}
\end{figure}

\begin{remark}
\label{rem:power-of-2}
  The computation of $X_a$ (described in the proof) proceeds by repeatedly considering the largest degree $d$ of a collector and assigning $d-1$ of the corresponding emitters to $X_a$ and one to $X_b$; it is therefore a greedy algorithm. The process is to be stopped whenever $X_a$ reaches half the size of $X_j^-$. 
  If $X_a$ exceeds this threshold during step $j$, then some emitters can be arbitrarily transferred from $X_a$ to $X_b$ to preserve the fact that $|X_{j+1}^-|$ is a power of two. The case that $|X_1^-|=k$ is not a power of two is addressed similarly after the first iteration, in order to set the size of $X_b$ to the highest power of two below $k$.
\end{remark}

\noindent{\bf How $X_a$ and $X_b$ are then used:} When an emitter $u$ in $X_a$ reaches another emitter $v$ in $X_b$, the corresponding journey arrives at $v$ through some edge $e^i(v)$ with $i \in {\cal I}_j$. We say that $u$ \emph{partially delegates} its emissions to $v$ in the sense that all collectors that $v$ can reach after this time can be reached from $u$ (\cref{fact:E2E}), the other collectors being possibly no longer reachable from $v$ after this time. 

\begin{lemma}
  \label{lem:direct-edges}
  If an emitter $u$ makes a partial delegation to $v$ in step $j$, then the number of collectors that $u$ may no longer reach through $v$ is at most $2^{j+3}-8$.
\end{lemma}

\begin{proof}
  This number is the largest value in the current interval; it corresponds to the largest rank of the edge through which the journey from $u$ may have arrived at $v$. Every edge whose rank locally to $v$ is larger than $2^{j+3}-8$ can still be used and thus the corresponding collectors are still reachable. This also implies that all other edges incident to $v$ (so with rank lower than $2^{j+3}-8$) may have already disappeared.
\end{proof}

\cref{lem:direct-edges} implies the following.

\begin{lemma}
  \label{lem:cost}
  In each step $j$, at most $2^{j+3}$ edges are selected relative to every eliminated emitter.
\end{lemma}

\begin{proof}
   Every eliminated emitter $u$ delegates to some emitter $v$. This is done through a 2-hop journey (composed of 2 edges). Additionally, through this delegation, $u$ may not be able to reach at most $2^{j+3}-8$ collectors, for which at most $2^{j+3}-8$ direct edges from $u$ are selected. The total amount of edges used for this partial delegation is thus at most $2^{j+3}-6 \leq 2^{j+3}$.
\end{proof}

More globally, let $J_j$ be the edges corresponding to all of the journeys used for partial delegation from vertices in $X_a$ to vertices in $X_b$ in step $j$, and $D_j$ the union of edges between the emitters in $X_a$ and the missed collectors. Let $S_j = J_j \cup D_j$. The algorithm thus consists of selecting all the edges in $S_j$ for inclusion into the spanner. Then $X^-_{j+1}$ is set to $X_b$ and the iteration proceeds with the next step.
The computation goes for $j$ ranging from $1$ to $\log_2 k - 3$, which leaves exactly {\em eight} final emitters alive. All the remaining edges of these emitters (call them $S_{last}$) are finally selected.
Overall, the final spanner is the union of all selected edges, plus the edges corresponding to the two initial matchings, \ie $S=(\cup_j S_j) \cup S_{last} \cup S^- \cup S^+$.

\begin{theorem}
  \label{th:layered-number}
  $S$ is a temporal spanner of the complete bipartite graph $B$ and it has $\Theta(n \log n)$ edges.
\end{theorem}
\begin{proof}[Proof]
  The key observation for establishing {\em validity} of the spanner is that eliminated emitters reach all collectors either directly or through an emitter that can still reach this collector {\em afterwards}. This property applies transitively (thanks to the disjoint and increasing intervals) until eight emitters remain, all the edges of which are selected for simplicity. Therefore, every initial emitter can reach all collectors. The rest of the arguments are the same as in the proof of Theorem~\ref{th:bidirectional}: all vertices in the input clique can reach at least one emitter $u$ through $e^-(u)$, and be reached by at least one collector $v$ through $e^+(v)$.
  
  Regarding the size of the spanner, in step $j$, $\frac{k}{2^{j}}$ emitters are eliminated and at most $2^{j+3}$ edges are selected for each of them (\cref{lem:cost}), amounting to at most $8k=4n$ edges. The number of iterations is $\Theta(\log k)=\Theta(\log n)$. Finally, the sets $S_{last}, S^-,$ and $S^+$ each contain $\Theta(n)$ edges (and $S^+$ is actually included in $S_{last}$). 
\end{proof}

\begin{theorem}
\label{theorem:nlogn_clique}
Simple temporal cliques always admit $O(n\log n)$-sparse spanners.
\end{theorem}
\begin{proof}
  In each recursion of the global algorithm, either the residual instance of the fireworks procedure is $2$-hop dismountable and the algorithm recurses on a smaller instance induced by a removed vertex, after selecting a {\em constant} number of edges, or the algorithm computes a $\Theta(n \log n)$-sparse spanner of the residual instance through the layered delegation process. Let $n_1$ be the number of times the graph is $2$-hop dismounted and $n_2 = n-n_1$ be the number of vertices of the residual instance when the layered delegation process begins (if applicable, $0$ otherwise). The $2$-hop dismounting contributes $\Theta(n_1)$ edges, and layered delegation contributes $\Theta(n_2 \log n_2)$ edges if it is performed, so the resulting spanner has $O(n \log n)$ edges.
\end{proof}

\section{Time Complexity and Summary of the Algorithm}
\label{sec:polynomial}

This short section reviews the cost in time of the main components involved in the algorithm. This discussion is by no means a detailed analysis, its purpose is rather to sustain the claim that the overall algorithm runs in polynomial time. The input to the algorithm is a temporal clique $\mathcal{G} = (G, \lambda)$, where $G = K_n$ is the complete graph on $n$ vertices, and $\lambda$ is a simple labelling of the edges, which may be represented as a permutation $\pi$ of $\{0, 1, 2, ..., {n \choose 2}-1\}$. 
The global algorithm is portrayed in \cref{fig:global_algorithm}.
Observe that whenever the algorithm recurses, the number of vertices is reduced by one, and in each recursion the fireworks process is run twice (forward and backward), possibly followed by the layered delegation processes, in which case the algorithm subsequently terminates. The fireworks process will thus run less than $2n$ times and the delegation process at most one time. The fireworks process first identifies the edges which are minimum (maximum) for at least one vertex. Then, it transforms this structure by means of a set of local operations consisting of flipping edges (at most once) or discarding them. As for layered delegations, the main operation is the composition of the delegation sets $X_a$ and $X_b$, which is done a logarithmic number of times by a greedy procedure whose main operation is to examine the degrees of all collectors to detect the local maximum among their labels. In light of these observations, we hope that it is clear to the reader that the overall running time is polynomial.

\begin{figure}
    \centering
    \includegraphics[scale=.7]{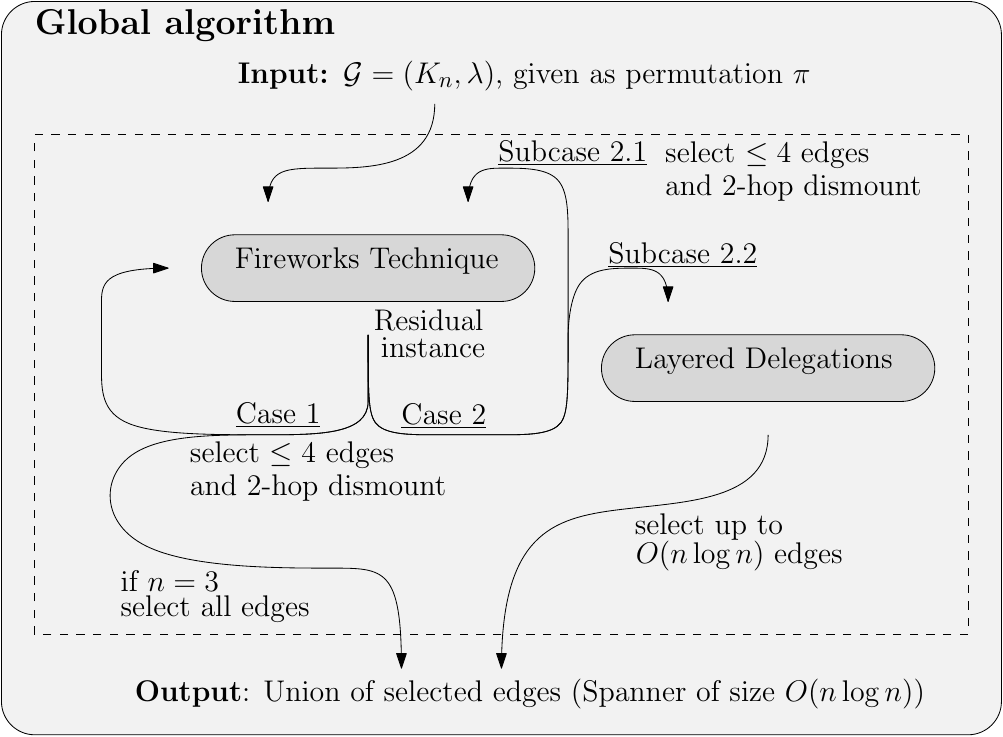}
    \caption{The global algorithm using the fireworks technique, dismounting, and layered delegations.}
    \label{fig:global_algorithm}
\end{figure}

\section{Open Questions}
\label{sec:tightness}

Despite extensive computer search, we were not able to find an instance (out of hundred of millions of instances of different sizes) which does not admit a spanner of either
$2n-3$ or $2n-4$ edges. The latter is actually a lower bound, as a classical result in gossip theory (see \eg Facts F29 through F32 in~\cite{HLPR13}) can be adapted to rule out the possibility that a graph with single labels is temporally connected with less than $2n-4$ edges.
Interestingly, we found many instances which are neither pivotable (see Section~\ref{sec:preliminary}) nor $k$-hop dismountable (see Section~\ref{sec:delegation}), and yet admit a spanner of size $2n-3$ or $2n-4$. This suggests further investigation to understand the structure of simple temporal cliques. In particular:

\begin{conjecture}
  \label{conj:linear}
  Do simple temporal cliques always admit temporal spanners of size $\Theta(n)$?
\end{conjecture}

\begin{conjecture}
  \label{conj:linear-strong}
  Do simple temporal cliques always admit temporal spanners of size $2n-3$?
\end{conjecture}

\noindent
A first step towards answering these questions might be to show that this is true in a probabilistic sense, which seems to be the case. More precisely, let a random simple temporal clique be built by assigning to every edge $e_i$ the label $\pi[i]$, where $\pi$ is a random permutation of $\{0,1,\ldots,{n \choose 2}-1\}$. 
Preliminary experiments suggest that random temporal cliques may asymptotically almost surely (\textit{a.a.s.}) be both pivotable and fully-dismountable (\textit{i.e.}, the probability that this is the case tends to $1$ as $n$ goes to infinity).  However, our experiments also suggest that there exist instances of arbitrary sizes which are neither pivotable nor dismountable.

\begin{conjecture}
  Do random simple temporal cliques \textit{a.a.s.} admit spanners of at most $2n-3$ edges?
\end{conjecture}

\begin{conjecture}
  Are random simple temporal cliques \textit{a.a.s.} fully ($k$-hop) dismountable? 
\end{conjecture}

\begin{conjecture}
  Are random simple temporal cliques \textit{a.a.s.} pivotable?
\end{conjecture}

\noindent On the deterministic side again, a natural question is whether a more general class than complete graphs deterministically admits sparse spanners, and what is the role of density. The family of counter-examples from Axiotis and Fotakis~\cite{AF16} have asymptotic density $1/9$, which leaves a significant gap between this family and complete graphs.

\begin{conjecture}
  Is there a larger class of dense graphs than complete graphs that always admit $o(n^2)$-sparse temporal spanners?
\end{conjecture}

Finally, we know that some graphs exist which are neither pivotable nor dismountable (see Figure~\ref{fig:neither} in Section ~\ref{apx:non-dismountable}), and some experiments that we have conducted suggest that arbitrarily large graphs that are both non-pivotable and non-dismountable may exist. However, the characterization of an infinite family with these properties is left open in this paper. 
\begin{conjecture}
  Is there a canonical way to construct graphs of arbitrary size which are neither pivotable nor dismountable?
\end{conjecture}

\section{Concluding remarks}
\label{sec:conclusion}

In this paper, we established that sparse temporal spanners always exist in temporal cliques, proving constructively that one can find $O(n \log n)$ edges that suffice to preserve temporal connectivity.
Our results hold for non-strict journeys with single or multiple labels on each edge, and strict journeys with single or multiple labels on each edge with the property that there is a subset of locally exclusive single labels. Our results give the first positive answer to the question of whether any class of dense graphs always has sparse temporal spanners.

To prove our results, we introduced several techniques
(pivotability, delegation, dismountability and $k$-hop dismountability, forward and backward fireworks, partial delegation, and layered delegations), all of which are original and some of which might be of independent interest. Whether some of these techniques can be used for more general classes of graphs is an open question. Delegation and dismounting rely explicitly on the graph being complete; however, refined versions of these techniques like partial delegation might have wider applicability.

One of the main open questions is whether sparse spanners exist in more general classes of dense graphs, keeping in mind that some dense graphs are unsparsifiable. Another is whether a better density than $O(n \log n)$ can be achieved in the particular case of temporal cliques, and more precisely can spanners of size $O(n)$ always be found in this case. Experiments that we have conducted suggest that spanners with $O(n)$ edges might always exist. At a deeper level, all these questions pertain to identifying and studying analogues of spanning trees in temporal graphs, which do not enjoy the same matroid structure as in standard graphs and seem much more complex.\\

\noindent {\bf Acknowledgements.} We thank the referees for their careful reading and constructive comments which significantly improved the presentation of these results.

\bibliographystyle{plain}
\bibliography{jcss}

\begin{thebibliography}{10}

\bibitem{AHU}
Alfred~V. Aho, John~E. Hopcroft, and Jeffrey~D. Ullman.
\newblock {\em Data Structures and Algorithms}.
\newblock Addison-Wesley, 1983.

\bibitem{AGMS17}
Eleni~C. Akrida, Leszek Gasieniec, George~B. Mertzios, and Paul~G. Spirakis.
\newblock The complexity of optimal design of temporally connected graphs.
\newblock {\em Theory of Computing Systems}, 61(3):907--944, 2017.

\bibitem{arya}
Sunil Arya, David~M Mount, and Michiel Smid.
\newblock Dynamic algorithms for geometric spanners of small diameter:
  Randomized solutions.
\newblock {\em Computational Geometry}, 13(2):91--107, 1999.

\bibitem{AE84}
B.~Awerbuch and S.~Even.
\newblock Efficient and reliable broadcast is achievable in an eventually
  connected network.
\newblock In {\em 3rd Symposium on Principles of Distributed Computing (PODC)},
  pages 278--281, 1984.

\bibitem{AF16}
Kyriakos Axiotis and Dimitris Fotakis.
\newblock On the size and the approximability of minimum temporally connected
  subgraphs.
\newblock In {\em 43rd International Colloquium on Automata, Languages, and
  Programming (ICALP)}, pages 149:1--149:14, 2016.

\bibitem{BFJ03}
Bin-Minh {Bui-Xuan}, Afonso Ferreira, and Aubin Jarry.
\newblock Computing shortest, fastest, and foremost journeys in dynamic
  networks.
\newblock {\em International Journal of Foundations of Computer Science},
  14(2):267--285, 2003.

\bibitem{CFQS12}
Arnaud Casteigts, Paola Flocchini, Walter Quattrociocchi, and Nicola Santoro.
\newblock Time-varying graphs and dynamic networks.
\newblock {\em International Journal of Parallel, Emergent and Distributed
  Systems}, 27(5):387--408, 2012.

\bibitem{restless}
Arnaud Casteigts, Anne{-}Sophie Himmel, Hendrik Molter, and Philipp Zschoche.
\newblock The computational complexity of finding temporal paths under waiting
  time constraints.
\newblock {\em CoRR}, abs/1909.06437, 2019.

\bibitem{CPS19}
Arnaud Casteigts, Joseph~G. Peters, and Jason Schoeters.
\newblock Temporal cliques admit sparse spanners.
\newblock In {\em 46th International Colloquium on Automata, Languages, and
  Programming (ICALP)}, pages 134:1--134:14, 2019.

\bibitem{Chew86}
Paul Chew.
\newblock There is a planar graph almost as good as the complete graph.
\newblock In {\em 2nd Symposium on Computational Geometry}, pages 169--177,
  1986.

\bibitem{diluna}
Giuseppe~Antonio Di~Luna, Stefan Dobrev, Paola Flocchini, and Nicola Santoro.
\newblock Distributed exploration of dynamic rings.
\newblock {\em Distributed Computing}, 33(1):41--67, 2020.

\bibitem{elkin}
Michael Elkin.
\newblock A near-optimal distributed fully dynamic algorithm for maintaining
  sparse spanners.
\newblock In {\em 26th ACM Symposium on Principles of Distributed Computing},
  pages 185--194, 2007.

\bibitem{erlebach}
Thomas Erlebach, Michael Hoffmann, and Frank Kammer.
\newblock On temporal graph exploration.
\newblock In Magn{\'{u}}s~M. Halld{\'{o}}rsson, Kazuo Iwama, Naoki Kobayashi,
  and Bettina Speckmann, editors, {\em 42nd International Colloquium on
  Automata, Languages, and Programming (ICALP)}, volume 9134 of {\em LNCS},
  pages 444--455. Springer, 2015.

\bibitem{gottlieb}
Lee-Ad Gottlieb and Liam Roditty.
\newblock Improved algorithms for fully dynamic geometric spanners and
  geometric routing.
\newblock In {\em 19th ACM-SIAM Symposium on Discrete Algorithms (SODA)}, pages
  591--600, 2008.

\bibitem{HLPR13}
H.A. Harutyunyan, A.L. Liestman, J.G. Peters, and D.~Richards.
\newblock Broadcasting and gossiping in communication networks.
\newblock In J.L. Gross, J.~Yellen, and P.~Zhang, editors, {\em Handbook of
  Graph Theory, Second Edition}, chapter 12.2, pages 1477--1494. CRC Press,
  2013.

\bibitem{holme}
Petter Holme and Jari Saram{\"a}ki.
\newblock {\em Temporal Networks}.
\newblock Springer, 2013.

\bibitem{KKK02}
David Kempe, Jon Kleinberg, and Amit Kumar.
\newblock Connectivity and inference problems for temporal networks.
\newblock {\em Journal of Computer and System Sciences}, 64(4):820--842, 2002.

\bibitem{matchings}
George~B. Mertzios, Hendrik Molter, Rolf Niedermeier, Viktor Zamaraev, and
  Philipp Zschoche.
\newblock Computing maximum matchings in temporal graphs.
\newblock In {\em 37th International Symposium on Theoretical Aspects of
  Computer Science (STACS)}, volume 154 of {\em LIPIcs}, pages 27:1--27:14.
  Schloss Dagstuhl - Leibniz-Zentrum f{\"{u}}r Informatik, 2020.

\bibitem{michail}
Othon Michail and Paul~G Spirakis.
\newblock Elements of the theory of dynamic networks.
\newblock {\em Communications of the ACM}, 61(2):72--72, 2018.

\bibitem{NS07}
Giri Narasimhan and Michiel Smid.
\newblock {\em Geometric Spanner Networks}.
\newblock Cambridge Univ. Press, 2007.

\bibitem{orda90}
Ariel Orda and Raphael Rom.
\newblock Shortest-path and minimum-delay algorithms in networks with
  time-dependent edge-length.
\newblock {\em Journal of the ACM (JACM)}, 37(3):607--625, 1990.

\bibitem{PS89}
David Peleg and Alejandro~A Sch{\"a}ffer.
\newblock Graph spanners.
\newblock {\em Journal of Graph Theory}, 13(1):99--116, 1989.

\bibitem{viard}
Tiphaine Viard, Matthieu Latapy, and Cl{\'e}mence Magnien.
\newblock Computing maximal cliques in link streams.
\newblock {\em Theoretical Computer Science}, 609:245--252, 2016.

\bibitem{Wilson75}
Richard~M. Wilson.
\newblock Decompositions of complete graphs into subgraphs isomorphic to a
  given graph.
\newblock In {\em 5th British Combinatorial Conference}, pages 647--659.
  Congressus Numerantium XV, 1975.

\bibitem{niedermeier}
Philipp Zschoche, Till Fluschnik, Hendrik Molter, and Rolf Niedermeier.
\newblock The complexity of finding small separators in temporal graphs.
\newblock In {\em 43rd International Symposium on Mathematical Foundations of
  Computer Science (MFCS)}, pages 45:1--45:17, 2018.

\end{thebibliography}

\end{document}